%% file: worldgraph.tex
\documentclass[12pt]{article}
\usepackage{graphicx}
\usepackage[latin1]{inputenc}
\usepackage{a4}
\usepackage{bbm}
\usepackage{amsmath}
\usepackage{amsthm}
\usepackage{amssymb}
\usepackage{eucal}
\usepackage{mathrsfs}
\usepackage{cite}
\usepackage[colorlinks=true]{hyperref}

\include{mydefinitions1}

\newcommand{\dhalf}{\ensuremath{{\frac{d}{2}}}}

\newcommand{\commentpaper}[1]{}

\editdisable{17.04.2007}
\newcommand{\ed}{\edit[17.04.2007]{EDIT}} 
\editenable{2ndrevision}

\begin{document}

\title{World Graph Formalism for Feynman Amplitudes}
\date{25th April 2007\\
1st revised version: 21st April 2008\\
2nd revised version: 19th August 2008}
\author{Helmut~H\"olzler}
\maketitle

\begin{abstract}
A unified treatment of Schwinger parametrised Feynman amplitudes is
suggested which addresses vertices of arbitrary order on the same footing
as propagators. Contributions from distinct diagrams are organised
collectively. The scheme is based on the continuous graph Laplacian. The
analogy to a classical statistical diffusion system of vector charges on the
graph is explored.
\end{abstract}

\include{schwinger}

\bibliographystyle{habbrv}
\bibliography{../hh}

\end{document}

%% file: mydefinitions1.tex
\usepackage{ifthen}

\DeclareMathOperator{\Tr}{Tr}

\DeclareMathOperator{\diff}{d}

\DeclareMathOperator{\Sym}{Sym}

\DeclareMathOperator{\Mod}{Mod}
\DeclareMathOperator{\Dom}{Dom}

\DeclareMathOperator{\AdS}{AdS}
\DeclareMathOperator{\hypergeombare}{F}

\newcommand{\AdSCFT}{AdS/CFT\;}

\newcommand{\pFq}[2]{{\sideset{_#1}{_#2}\hypergeombare}}
\newcommand{\hypergeom}{\pFq{2}{1}}
\newcommand{\dif}[2][noparam]{\ifthenelse{\equal{#1}{noparam}}{\diff \! #2 \,}{\diff^{#1} \! #2 \,}}

\newcommand{\Em}{\ensuremath{{\mathbbm{E}}}}

\newcommand{\wick}[1]{:\negmedspace#1\negmedspace:}

\newcommand{\cev}[1]{\ensuremath{\overleftarrow{#1}}}
\newcommand{\onehalf}{\ensuremath{{\frac{1}{2}}}}

\newcommand{\threehalf}{\ensuremath{{\frac{3}{2}}}}

\theoremstyle{plain}
\newtheorem{prop}{Proposition}
\newtheorem{thm}[prop]{Theorem}
\newtheorem{lem}[prop]{Lemma}

\theoremstyle{definition}

\newtheorem*{defn*}{Definition}
\newtheorem*{ass*}{Assumption}
\newtheorem*{expl*}{Example}

\theoremstyle{remark}

\newtheorem*{rem*}{Remark}

\newcommand{\comment}[1]{}

\comment{This is a comment and it won't show up in the output even if
  it spans more than one line of the input file or the line breaks
  move around.}

\newcommand{\condcomment}[2]{\ifthenelse{#1}{#2}{}}

\newcommand{\editenable}[1]{\newcounter{#1} \setcounter{#1}{1}}
\newcommand{\editdisable}[1]{\newcounter{#1} \setcounter{#1}{0}}
\editenable{genericedit}
\newcommand{\edit}[2][genericedit]
{\ifthenelse{\equal{\value{#1}}{0}}{}{\P\marginpar{\small #2 \arabic{#1}}\addtocounter{#1}{1}} }

%% file: schwinger.tex
\section{Introduction}

Recently, the well-known technique of Schwinger parametrisation
of Feynman diagrams \cite{Itzykson1980} has received renewed interest, funneled by the
speculation of Gopakumar \cite{Gopakumar:2003ns, Gopakumar:2004qb,
  Gopakumar:2005fx} that it might be the key to an understanding of the
$\AdSCFT$ correspondence on the diagrammatic level of correlation
functions. This suggestion is based on the observation that the
``lifting'' of free large N U(N) symmetric gauge field theory amplitudes
of twist-2 operators
from the boundary into bulk $\AdS$ space has a very
natural appearance when one applies the Schwinger parametrisation to
the boundary amplitudes, at least in the simplest nontrivial case of
three-point amplitudes. The particular integral
representation of the bulk amplitudes obtained in this manner is being
interpreted as a string theory on a curved space in the limit of
large curvature; and since such a theory is currently
beyond a direct understanding, these results consequently incited a
program of trying to gain knowledge about this particular
string theory by the study of the lifted boundary field
theory.
According to the advertised model, the correspondence proceeds in two
clearly distinct levels: First, there should be a correspondence between the
boundary amplitudes and an open string theory including branes; second, by
open-closed duality, these open string amplitudes should be equivalent to
closed string amplitudes living in the bulk (see eg \cite{Ooguri2002}).

In this paper, we concentrate on the
Schwinger parametrisation of quantum field theories.
In its simplest form, it is obtained by going to the momentum space
representation of a Feynman diagram derived from the path integral and
rewriting it, making use of the representation
\begin{equation}
\label{eqn:schwinger}
\frac{1}{q^2 + m^2}
= \int_0^\infty \dif\tau e^{- \tau (q^2 + m^2)}.
\end{equation}
The issue that Schwinger parametrisation can be interpreted as
being generated by ``world line'' path integrals is rather settled by
now \cite{Schubert:2001he, Strassler:1992zr}. In section 8 of Schubert's review \cite{Schubert:2001he},
the question of how to treat multi loop Feynman graphs in that context is discussed, noting
that the world line formalism cannot be implemented immediately on those
graphs since -- in opposition to one-loop graphs -- the graph
cannot be treated as a differentiable manifold (parametrised $S^1$), due to the vertices.
The solution offered is rather a pragmatic one: Multi-loop Feynman graphs
are constructed from a one-loop ``spider'' graph with several external insertions,
by connecting some of the external insertions by propagators in Schwinger
parametrisation. The resulting amplitudes are then manipulated algebraically
and the loop momenta are integrated out.

It is the intention of this text to suggest a direct, stringent procedure
implementing the world line formalism also for graphs with vertices,
without resorting to iterative construction out of simpler graphs.
We show how a multi loop Feynman graph can be treated as one-dimensional
manifold with branching points, enabling us to write down a ``world graph''
formalism which delivers equivalent results
(e.g., formula \ref{eqn:G:Zeff} on page \pageref{eqn:G:Zeff})
in a consistent, ``one-step'' fashion.
It takes the form of a simple diffusion path integral mapping the
\textit{complete} Feynman graph into
coordinate space, with particular continuity conditions at the
vertices. The main result reported in this article is the methodical derivation of this new route.
A side result is to offer a more detailed view on the role of generalised
Schwinger parameters, or ``moduli'', of Feynman graphs.

The formalism is not restricted to a particular set of
vertices, or particle types (although it is rather natural to employ
it for massless particles). From this point of view, Schwinger
parametrisation is a notion which makes sense
for a Feynman graph as a whole - one should rather speak of
``world graphs'' than world lines. We will show there is a close
connection to the interpretation of Feynman amplitudes as a
partition sum of charged particles residing on the graph, generalising a concept which has
successfully been applied to one-loop and two-loop graphs.
It is crucial that these partition sums are in fact sums over all
different possibilities to connect the external propagators to the graph.
Polynomial prefactors in the internal momenta of the Feynman
amplitudes, ie derivative interactions, for vector or tensor particles can be included easily by
introducing infinitesimal ``test-dipoles'' on the graph.

Let us mention at this point the connection to string theory:
Bern and Kosower \cite{Bern1992} have shown in a long work that
Schwinger parametrised amplitudes can be obtained
from the infinite tension limit of a certain open string theory,
where the strings degenerate into point particles.
The tachyonic modes of the string can be
made to decouple, and the only excitations left in this
limit are the massless modes (all other modes become too massive to be
excited at all). The Feynman rules which result in this limit come out
very naturally in the Schwinger parametrised form (however see eg \cite{Roland1996}
for a different limit retaining only the \textit{tachyonic} modes,
producing scalar $\phi^3$ theory). In fact, we can say
more: The theory of the massless vector fields obtained in this way is a Yang-Mills
theory; if the strings carry Chan-Paton factors, then it is a
non-Abelian gauge theory \cite{Polchinski1998}. Now, in the usual Feynman diagrammatic
calculation of amplitudes in non-Abelian gauge theories, there is a
lot of redundancy: The amplitudes corresponding to the diagrams
consist of \textit{very} many different summands, and there occurs a
host of cancellations between those, so that the
final result usually reduces to a comparably compact expression.
When the same amplitudes are derived by way of
the infinite-tension limit of string theory, they
turn out to be very well organised so that cancellations are immediate
\cite{Bern1991, Lam1994}.

Let us mention another interesting detail: On the string theory side,
we have to integrate over the so-called ``string moduli''. These are
parameters which label uniquely the conformally inequivalent ways to
put a metric on the string world sheet. Taking the
infinite-tension limit, the string moduli are mapped partially onto the
Schwinger parameters. A point we want to stress is that the
mathematical problems which are a major obstruction when one tries to
consider more complicated string world sheet topologies in the
infinite tension limit are understood rather naturally in the
world graph limit.

We thank the referee for pointing out to us an
earlier work by Dai and Siegel \cite{Dai2007}, who explore a related
approach to multi-loop amplitudes. As a starting point, they choose the first-quantised formalism,
developed by Strassler \cite{Strassler:1992zr} and many others, which includes
an integral over the reparametrisation group of the parametrised Feynman graphs,
and requires the subsequent fixing of this reparametrisation symmetry.
Their conclusions are similar, stressing along the way the ``electrical analogy''
which is obtained when the momentum flux through the diagram is set in analogy
to a (vector) current. Explicitly, their approach is spelled out only for scalar
fields.

The organisation of this paper is the following:
In section 2, we introduce the na\"ive Schwinger representation
and show how for each propagator it can be interpreted as
a diffusion kernel, implying the world line picture. In section 3, we
introduce the world graph formalism, enlarging the diffusion scheme to
complex graphs, and state its main content as a theorem.
In section 4, we show the equivalence of the
world graph scheme to the partition sum of a system of charged
particles residing on the graph and complete the proof of the theorem;
we give some elementary examples of the technique.
Finally, in the remaining section, we extend the formalism to vector
and tensor particles.

\section{World-line formalism}
\subsection{Schwinger parametrisation}

The Schwinger parametrisation of the correlation functions of a Lagrangian field
theory in $d$-dimensional Euclidean space containing a set of scalar fields
and an arbitrary non-derivative polynomial interaction is based on the
perturbative expansion of the effective action in momentum space.
The effective action is the sum over connected, amputated Feynman diagrams,
containing massive propagators
\begin{equation}
\label{eqn:internal}
 G_m(q) = \frac{1}{q^2 + m^2}
\end{equation}
and vertices with varying coordination number $n$, carrying a momentum conserving factor
$
  - \frac{c_n}{(2\pi)^{\left(\frac{n}{2} - 1\right)d}}
   \delta^{(d)}\Big( \sum_{j=1}^n q_j \Big),
$
where $c_n$ is the coupling and $q_j$ are the incoming momenta.

To a vertex $v$, we assign the external momentum
$k_v$, the total sum of the joint momenta entering
the diagram through \textit{all} external, amputated legs of $v$.
External legs are thus effectively represented by vertices with
non-conservation of momentum of the internal
propagators. Conversely, internal vertices will be
treated as being connected to imagined external legs with zero
momentum entering the graph.
Finally, the internal momenta are integrated over.

The na\"ive Schwinger representation is obtained by blindly representing
each internal propagator (\ref{eqn:internal}) by formula (\ref{eqn:schwinger}).
We have thus for each propagator $j$ a Schwinger modulus (Schwinger
parameter) $\tau_j$. As a result, loop momenta integrations are
Gaussian and can be performed explicitly, leaving only the integrals over the
Schwinger parameters. The
result is the well-known formula which reduces to a certain sum of
``two-trees'' of the graph (eg \cite{Lam1993}); the precise form is
irrelevant here.

\comment{
If we consider a theory containing particles with spin, then the
propagators will contain additional factors which are basically
polynomials in the propagator momenta; at the vertices, the various
momenta of adjacent propagators are contracted.
}

Formula (\ref{eqn:schwinger}) has an interpretation based on the
diffusion equation. Each vertex $v$ obtains an additional coordinate
$x_v \in \mathbb R^d$; momentum conservation at $v$
is represented by the integral
$\delta^{(d)}\Big( \sum_{j=1}^n q_j \Big)
= (2\pi)^{-d} \int \dif[d]{x_v} \exp - \left(i x_v \cdot \sum_{j=1}^n q_j \right)$.
A Schwinger parametrised propagator running from vertex $x_1$ to vertex $x_2$ evaluates to
\begin{multline}
\label{eqn:alpha:propagator}
G_m(x_1-x_2)
= \int_0^\infty \dif{\tau}
  \int \dif[d]q
  \exp \left( -i(x_1-x_2) \cdot q - \tau (q^2 + m^2) \right)\\
= \pi^{d/2} \int_0^\infty \dif{\tau} \frac{1}{\tau^{d/2}}
  \exp - \left( \frac{(x_1-x_2)^2}{4 \tau} + m^2 \tau \right).
\end{multline}
If there are external legs carrying momentum $k_v$ attached to
vertex $v$, then we are left with the factor
\begin{equation*}
\exp -( ix_v \cdot k_v) .
\end{equation*}
\begin{figure}
\centering \includegraphics[width=3cm]{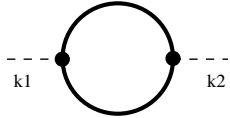}
\caption{Self energy diagram of $\phi^3$-theory. Dashed lines:
  amputated legs. Momenta are incoming.}
\label{fig:phi3-se}
\end{figure}
As an example, the one-loop self
energy of scalar $\phi^3$ theory (cf. fig. \ref{fig:phi3-se}) is
\begin{align}
\label{eqn:one-loop-se}
I(k_1, k_2)
=& \frac{1}{2}\cdot \frac{c_3}{(2\pi)^{\frac{3d}{2}}}
   \int \dif[d]x \exp \left( -ix\cdot k_1 \right) \cdot
   \frac{c_3}{(2\pi)^{\frac{3d}{2}}}
   \int \dif[d]y \exp \left( -iy\cdot k_2 \right) \nonumber\\
 &\pi^{d/2} \int_0^\infty \dif{\tau_1} \frac{1}{\tau_1^{d/2}}
  \exp - \left( \frac{(x-y)^2}{4 \tau_1} + m^2 \tau_1 \right) \\
 &\pi^{d/2} \int_0^\infty \dif{\tau_2} \frac{1}{\tau_2^{d/2}}
  \exp - \left( \frac{(x-y)^2}{4 \tau_2} + m^2 \tau_2 \right).\nonumber
\end{align}
By shifting $y \rightarrow y + x$, we eliminate $x$ from
the quadratic exponent; the $x$-integration is then seen to represent
external momentum conservation.

The propagator (\ref{eqn:alpha:propagator}) is related to the
well-known Gaussian kernel for the Wiener path integral describing
the diffusion of a particle for a time $\tau$
in $d$ (Euclidean) dimensions. We can write it formally as
\begin{align}
\label{def:WienerMeasure}
\mathscr K_\tau(x_1, x_0)
=& Z_0^{-1} \int_{\substack{x(\alpha)=x_1\\x(0)=x_0}} \mathscr Dx(t) \exp
     - \int_0^\alpha \dif{t} \left\{ \frac{1}{4} \dot{x}^2 \right\}\nonumber\\
=& \left( \frac{1}{4\pi \alpha} \right)^{d/2}
\exp - \frac{(x_1 - x_0)^2}{4\tau}.
\end{align}
The prefactor $\frac{1}{4}$ in the exponent is chosen in
order to agree with the standard literature, eg
\cite{Strassler:1992zr}. The kernel is normalised to
\begin{equation}
\label{eqn:Knorm}
\int \dif[d]{x_1} \mathscr K_\tau(x_1, x_0) = 1.
\end{equation}
Comparing with (\ref{eqn:alpha:propagator}), we find that the
propagator for a scalar quantum field can be expressed through
the heat kernel by
\begin{align}
\label{eqn:Gm:masssquare}
G_m(x-y)
=& \pi^{d/2}
   \int_0^\infty \dif{\tau} \frac{1}{\tau^{d/2}}
   \exp - \left( \frac{(x-y)^2}{4 \tau} + m^2 \tau \right) \nonumber\\
=& (2\pi)^d \int_0^\infty \dif{\tau}
   \exp \left( - m^2 \tau \right)
   \mathscr K_{\tau}(x-y).
\end{align}
One could interpret the exponential prefactor
as a dissipative term absorbing the diffusing
particle. The particle mass $m$
enters only in the dissipative part, and for $m=0$, we are
left with the dissipation-free heat kernel. This formula demands
$[\tau] = L^2$, so we must be cautious when we interpret $\tau$ as ``time'' (the reason
is that we have left out the diffusion constant).
The upshot is that the Euclidean field theory propagator is obtained by
averaging over diffusion ``times'' $\tau$,
with a weight factor falling off exponentially.
The diffusion picture does not extend to cover vertices, as it is.

In this approach, the Schwinger parametrised form of the
correlation function reads for any mass
\begin{multline}
\label{eqn:G:massless}
G_{\mathscr G}(k_1, \dots, k_n)
=  \frac{(2\pi)^{\frac{nd}{2}}}{\Sym(\mathscr G)}\,
   \Big( \prod_{\mathrm{vertices}\; v} - c_v
   \int \dif[d]{x_v} e^{-i k_v \cdot x_v}\Big) \\
   \prod_{\mathrm{propagators}\; j}
   \int_0^\infty \frac{\dif{\tau_j}}{(4\pi\tau_j)^{d/2}}
   \exp - \left( \frac{(x_{1(j)}-x_{2(j)})^2}{4 \tau_j} + m^2 \tau_j \right).
\end{multline}
Here, $\Sym(\mathscr G)$ is the symmetry factor of the graph $\mathscr G$, and
all powers of $2\pi$ at the vertices have been cancelled against the
propagators.

\subsection{Conformal propagators}
\label{sec:conf}

The Schwinger parametrisation is useful also to represent conformal
propagators
\begin{equation}
G_\Delta(x-y) = \frac{1}{|x-y|^{2\Delta}}.
\end{equation}
The scaling behaviour will be contained solely in a $\tau$-dependent
prefactor. Introducing a Schwinger-like integral representation
in coordinate space
\begin{equation*}
\frac{1}{(x^2)^\Delta}
= \frac{1}{\Gamma(\Delta)} \int_0^\infty \dif{\alpha}
  \alpha^{\Delta-1} e^{-\alpha x^2},
\hspace{1cm}\Re \Delta > 0,
\end{equation*}
we can compute the Fourier transform as
\begin{align*}
\int \dif[d]x e^{-i q \cdot x} \frac{1}{(x^2)^\Delta}
=& \frac{1}{\Gamma(\Delta)}
   \int \dif[d]x e^{-i q \cdot x}
   \int_0^\infty \dif{\alpha} \alpha^{\Delta-1} e^{-\alpha x^2}\\
=& \frac{\pi^\dhalf}{\Gamma(\Delta)}
   \int_0^\infty \dif{\alpha} \alpha^{\Delta-\dhalf-1}
      e^{-\frac{q^2}{4 \alpha}}
\end{align*}
by completing the square. Substituting $\alpha \rightarrow
(4\tau)^{-1}$, we get the usual Schwinger parametrisation
\begin{equation}
\label{eqn:Schwinger:conf}
G_\Delta(q) = \frac{2^{d - 2 \Delta} \pi^{\frac{d}{2}}}{\Gamma(\Delta)}
     \int_0^\infty \dif\tau \tau^{\frac{d}{2} - \Delta - 1}
         e^{-\tau q^2}.
\end{equation}
This representation is special insofar as the exponential part takes exactly
the form of a massless propagator. The only modification is the
power of $\tau$ in the Schwinger kernel. If $\Re \Delta < \dhalf$, we can evaluate the
integral explicitly to obtain
\begin{equation}
G_\Delta(q) = \frac{2^{d - 2 \Delta} \pi^{\frac{d}{2}}
                \Gamma(\frac{d}{2} - \Delta)}{\Gamma(\Delta)}
     |q|^{2\Delta - d}.
\end{equation}
Note that even when $\Delta$ is not within the bounds indicated, we
may by analytic continuation reach almost every complex $\Delta$.

\section{World-graph formalism}

Bosonic string theory can be formulated as a theory of $d$ scalar
``coordinate'' fields living on the two-dimensional string
world sheet (a Riemann surface which can have an arbitrary
topology). In the infinite string tension limit, strings are
effectively reduced to point-like particles (string shrunk to zero
length). Only very few string excitations survive
this limit, and it has been shown that it can
be consistently treated as a field theory.

The Feynman graphs of the limiting field theory may be quite literally interpreted as
the shrivelled remains of the string world sheet under infinite tension.
As they are one-dimensional, it has become
customary to refer to the propagators in this context as ``world lines'';
they are sewn together at the vertices.
On the technical level, scattering amplitudes mediated by
string interactions are turned into
a Schwinger parametrised version of field theoretic perturbation theory.
From a world sheet point of view however, the discrimination between
``free strings'' and ``vertices'' is artificial; given a section
of the world sheet, the question of whether it is a part of a
``vertex'' or not does not make sense at all.

Let us reverse the argument and ask whether there is an approach to
field theory which resolves the special treatment of the vertices.
This approach is found in the generalisation of Schwinger parametrisation.
Instead of using ``world lines'', we will employ
the concept of a ``world graph'' -- a Feynman graph
$\mathscr G$ is treated as a manifold with branching points at
the vertices which is mapped into the ambient space $\mathbb R^d$.
The \textit{world graph path integral} is a
weighted integral over all allowed embeddings of this kind.
No longer enter external particles the world graph at
``special points'' (vertices); rather, they are implemented by
sliding ``operator insertions'' which have to be taken care of in the
path integral.

At the branching points (vertices) of the
world graph, it will be required to impose a continuity condition.
It is precisely this continuity condition which marks the
difference between world line and world graph formalism.
The representation of the propagators is taken over unmodified from the
world line formalism: We associate a ``length'' $\tau_j$ to each
propagator in the graph; the final amplitude is obtained by performing
the world graph path-integral and integrating over all lengths with
the appropriate weight factor. The lengths $\tau_j$ will be called
``moduli'', in reference to the term used in string theory, where the moduli
characterise different conformal equivalence classes of Riemannian metrics on
the string world sheet. The parameter space for the moduli is
the \textit{moduli space} $\Mod[\mathscr G]$.
The ``dimension'' $\dim [\mathscr G]$ of a graph $\mathscr G$ is by
definition the dimension of its moduli space $\Mod[\mathscr G]$.
It is important to recognize that the positions $t_j$ of
the operator insertions are part of the moduli.
Graphs which are assigned such Schwinger lengths will be called
``metric'', in distinction to the usual Feynman graphs without
Schwinger lengths, which will be called ``non-metric''.

\subsection{Equivalence classes and cells of moduli space}

In the usual Schwinger parametrised Feynman graphs (world line
formalism), the order of the external (amputated) legs entering the graph is
fixed. The integration of the Schwinger parameters varies just the
distances in-between them; this is equivalent to
letting the insertions slide over the branches of the graph, without
changing their order, and integrating over all possible branch
lengths.

The world graph formalism does not admit a special treatment
of vertices; so we should expect that it makes sense to integrate over
\textit{all} admissible localisation points and
orders of the insertions.
When we integrate over moduli space and let the operator insertions
slide over the diagram, we find a natural sum over different orderings
of the external legs on the same branch of the diagram, and also
include the case where the external legs are inserted onto different
branches (see fig. \ref{fig:slide}).
\begin{figure}
\centering \includegraphics[width=11cm]{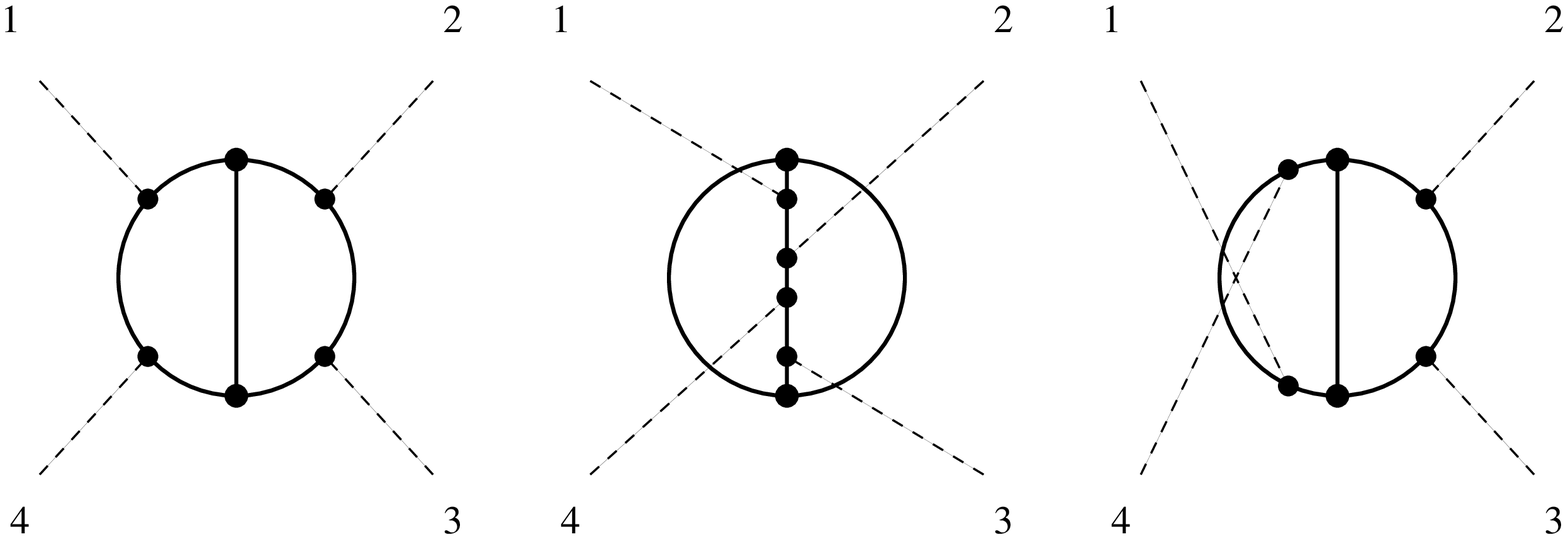}
\caption{Three different graphs grouping into the same equivalence class under
  class scheme A.}
\label{fig:slide}
\end{figure}
This clearly means that using the world graph formalism, we
necessarily will obtain \textit{sums over different Feynman graphs}.
Each sum defines a subset of the total set of all admissible Feynman
graphs (where we accept the Feynman rules as a priori given),
and because the world graph formalism should be equivalent to
the usual sum over Feynman graphs, it is important that each graph is
part of precisely one such sum. In other words, the set of all Feynman
graphs is partitioned into (mutually exclusive) equivalence
classes. The class containing the connected graph $\mathscr G$ will be denoted
$[\mathscr G]$; it is a set of non-metric connected Feynman graphs.
The moduli space cells $\Mod[\mathscr G]$ are in fact
parametrising these equivalence classes, rather than just single
graphs. Given a class $\mathfrak g \equiv [\mathscr G]$ and a
particular set of moduli $\tau \in \Mod \mathfrak g$, the particular
metric graph identified by the moduli will be denoted
$\mathfrak g(\tau)$. The moduli space $\Mod \mathfrak g$ is a measure space
of dimension $\dim \mathfrak g$, and the moduli $\tau$ are coordinates
on this space. The measure on $\Mod \mathfrak g$ is 
derived from the proposed equivalence to the usual
diagrammatic computation.

We will suggest a set of rules telling us how, given a graph
$\mathscr G$, we can subsequently generate all $\mathscr G'
\in [\mathscr G]$ in the equivalence class; it is easy to
show that these relations are reflexive, symmetric and transitive, and
therefore define a true partitioning into equivalence classes.
Let us point out that there are several consistent ways to define the
classes; we discuss them in turn.

As a prerequisite, we have to classify the operator insertions: Each insertion is
connected to a number of external legs entering the diagram, its
``external valency''; in the simplest case, it will be one. We
assume that all external legs are distinguishable. Likewise,
there is a number of internal legs connected to the insertion (the
``internal valency''). The sum of external and internal valency is the
total valency of the insertion, and this settles the coupling needed
at the insertion (ignoring the question of different particle types
for the moment). \textit{We stress once more that ``ordinary''
  vertices are treated throughout as operator insertions with external
  valency 0 in the world graph formalism.}

\begin{defn*}[Equivalence Class A]
Given a graph $\mathscr G$, the equivalence class $[\mathscr G]$ is
generated by letting insertions with internal valency 2 slide over the
complete graph, changing their order as they go along. Insertions with
internal valency other then 2 (1 or larger than 2) cannot slide;
however, if there are several insertions with identical valencies
(internal \textit{and} external), then the external legs attached to these
insertions may be \textit{permuted} groupwise; ie if several external legs
are attached to the same insertion in $\mathscr G$, then they must be
attached to the same insertion for every graph in $[\mathscr G]$
(external legs are ``sticky'').

The continuous moduli are thus given by coordinates of the sliding
insertions on the graph, and by the metric of the underlying ``torso''
\footnote{Symmetries will be discussed below.}; in addition, there are
discrete moduli counting possible permutations of the external legs.
\end{defn*}
This definition may seem a little arbitrary; however, it is the one which
is best suited to the statistical analogy which will be introduced below.
An example of discrete moduli is given in fig. \ref{fig:discrete}.
We will nevertheless write the integral over moduli space as $\int
\dif[{\dim [\mathscr G]}]{\tau}$.
\begin{figure}
\centering \includegraphics[width=10cm]{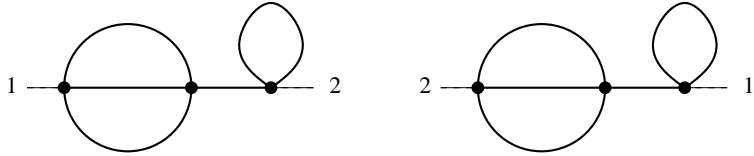}
\caption{External legs 1 and 2 can be interchanged without modifying the
  valency of the insertions. There is a discrete modulus
  (eg $\tau_{12} \in \{ 12, 21 \}$).}
\label{fig:discrete}
\end{figure}

There are further possibilities: A rather strict one is
\begin{defn*}[Equivalence Class B]
Given a graph $\mathscr G$, the equivalence class
$[\mathscr G] = \{ \mathscr G \}$ is minimal.
There are no permutations or rearrangements.
The corresponding moduli space $\Mod[\mathscr G]$ is made up of the
usual Schwinger parameters.
\end{defn*}
This is the class concept of the world line formalism. Finally, we may
be very liberal and make
\begin{defn*}[Equivalence Class C]
Given a graph $\mathscr G$, the equivalence class $[\mathscr G]$
contains all graphs which can be assembled by cutting all
internal propagators of $\mathscr G$ and
reconnecting the remaining insertions (retaining their internal valency)
in an arbitrary manner, under the constraint that the resulting
graph is connected.

The moduli space counts all different possible topologies, and on each
topology there are continuous moduli controlling the metrics.
\end{defn*}
Note that this includes ``discrete permutations''; the external legs
are still sticky. It is easy to see that the loop number
is constant within a class; the remaining class invariants are derived
from the types of operator insertions.
This latter class concept is the one which
is most closely related to closed string theory: A closed string world
sheet with finite tension is topologically characterised solely by its
loop number. Incidentally, it is the one which is
obtained following the approach of Bern and Kosower
\cite{Bern1992}.

In a wider sense, the moduli space of all possible graphs is made up
of different ``cells'' $\Mod[\mathscr G]$. In this
way, the total moduli space (containing all graphs) has a natural cell
structure \cite{Gopakumar:2005fx}. The class concepts introduced
regulate the extent of these cells. It is obvious that $B \leq A \leq C$, ie
B defines a subpartitioning of A, and A subpartitions C. The
partitioning A seems to be the one which incorporates systematically
the world graph concept without making the classes unnecessarily
large. This should not be mistaken for a physical statement: It is
rather one of convenience.

\ed On the other hand, the partitioning C is the only one treating
internal and external valencies alike. This can be seen by
following example: Consider two vertices $v$ and $w$ sliding along
one branch of a graph $\mathscr G$; assume that these vertices are
distinguishable (ie by their total valency); for concreteness, assume
that vertex $v$ is 3-valent and vertex $w$ 4-valent. We have to
decide what to do with the one remaining leg of $v$ and
the two remaining legs of $w$. If the remaining
legs of the vertices are all external (and thus amputated), then $v$ and $w$ may
exchange their order on the branch within $[\mathscr G]$ by
A. However, if the remaining legs are connected internally by a propagator to some
other vertex of the graph, they may not, by A. This is not so in C: There,
diagrams are completely rewired, and every ordering is included.

For the purpose of the examples given in this text, the concept A is
broad enough. The generalisations to C are immediate in most
cases. For this reason, \textit{we will in the rest of this text
  adhere to the class concept A}.

\paragraph{Symmetry factors.}
Each Feynman graph has to be divided by a symmetry
factor which is obtained in the usual way from the perturbative expansion
(it is the size of the automorphism group of the graph).
These factors are identically taken over in the world graph formalism.
If the topology of the underlying diagram varies
within a cell of moduli space, the symmetry factor may change.
For an important practical aspect, see however the comments
at the end of section \ref{sec:wp}.

\paragraph{Particle types.}
If there are different sorts of particles involved, then every
propagator has to be assigned a particle type. Two graphs of the same topology are
different by definition if the particle types are not completely identical.
The particle types of external, amputated legs are fixed by assumption.
In this case, we make the agreement that all possible assignments of
particle types to the internal propagators which can be satisfied by
a set of vertices from the Feynman rules are part of the class.
As we generate the graphs in a class, each time the topology changes
or one insertion crosses another, we have to change the particle type
of the propagators. This implies that the couplings have to vary as well.
There will be topologies that can be fulfilled (because the Feynman
graph corresponding to this particular ordering can be constructed
from the couplings) and topologies that will fail (because there
is no corresponding Feynman graph). There have to be additional,
discrete moduli to keep track of particle types and couplings.

As far as we allow arbitrary sets of vertices and
particles, this is already the end of the story. When the
analogy to string theory is deepened however, we expect that
string theory puts serious restrictions on the possible types of
particles and vertices, and their coupling constants.
The simplification of the amplitudes which has been
mentioned in the introduction should be present only for
very particular theories, and supposedly the
field theories motivated by string theory are strong candidates
here. Eg, we expect such simplifications for non-Abelian gauge
theories. The choice of a suitable class concept is at the heart of
these supposed simplifications.

Although, if we take serious the string theory parallelism, we should
only consider massless propagators, massive particles can without
problems be included into the scheme, with certain qualifications.
If all propagators have the same mass, then the mass prefactor is
trivially given by the total length of the graph.
If the propagators carry different masses, we
will not be able to give such a concise description: The operator
insertions describing external legs change the ``phase''
(mass) of the world graph lines. A possible way out is the inclusion
of a further ``particle type'' or mass field $m(t)$ on the world graph
and to describe the operator insertions as symmetric matrices
connecting different ``mass'' spaces (see section \ref{sec:mass}
below).

\subsection{Formulation of world-graph path integral}
\label{sec:wp}

Let us recapitulate: To compute the amplitude corresponding to an
equivalence class of Feynman graphs $\mathfrak g$, we first select one
particular metric graph $\mathfrak g(\tau)$ by a choice of Schwinger
parameters $\tau \in \Mod \mathfrak g$. On $\mathfrak g(\tau)$, we put
a theory of $d$ Euclidean massless scalar fields
$x: \mathfrak g(\tau) \rightarrow \mathbb R^d$ whose
dynamics is described by a diffusion (Wiener) process.
External legs are treated as operator insertions $e^{-i k_v \cdot x(t_v)}$;
and after integrating the fluctuations of the $d$ coordinate fields
$x(t)$, we have to integrate over the moduli space $\Mod \mathfrak g$.
This contains an integral over all possible positions of the
operator insertions on $\mathfrak g(\tau)$ as well as all possible
Schwinger lengths of the propagators with
the appropriate measure; there is a sum over the discrete moduli,
controlling permutations of the external legs, and finally over the
different topologies in the class $\mathfrak g$.

Note the important distinction between \textit{points on the
  graph $\mathfrak g(\tau)$} as a (singular) manifold which will be denoted
by small Latin letters $s,t,v \in \mathfrak g(\tau)$ and continuous ``moduli''
$\tau, T \in \mathbb R_+$, denoting \textit{distances or lengths} on
the graph. The location of the operator insertions is
determined by the moduli. By choosing a coordinate system
(parametrisation) of the graph resp. the propagators, a point
$t \in \mathfrak g(\tau)$ can sometimes be assigned a number - its coordinate. The
moduli, on the other hand, are independent of a choice of coordinates.

Instead of developing step-by-step the world graph formalism, we will
state at once the respective form of the world graph path integral and
prove subsequently that the amplitude obtained in this way is indeed
identical to direct computation by the usual Feynman rules.
Without loss of generality, we will study only connected graphs in the sequel.
We need some technical tools to begin with.

Introduce for vector-valued functions
$f,g : \mathfrak g(\tau) \rightarrow \mathbb R^d$ a real scalar product
\begin{equation}
\langle f, g \rangle
\equiv \int_{\mathfrak g(\tau)} \dif{t} f(t) \cdot g(t).
\end{equation}
It is similarly defined for scalars. This product defines a
real Hilbert space of functions on the metric graph $\mathfrak g(\tau)$.

The \textit{graph Laplacian} is an operator acting on functions defined
on the graph as a one-dimensional manifold with branching points at
the vertices
\footnote{These graphs have recently been termed ``quantum graphs'' in the
physics community. For an introduction and overview, see
\cite{Kuchment2004} and references therein.}.
This is not the discrete graph Laplacian;
we define the Laplacian $\triangle = \partial_t^2$
for functions on the graph as the one-dimensional continuous Laplacian along the
parametrised links of the graph; at the vertices, we get a
distributional contribution
\begin{equation}
\label{eqn:Laplacegraph}
\triangle f(t)
= \sum_{\mathrm{vertices}\; v}
  \left(\sum_{\mathrm{adjacent\;links\;}l} \lim_{(\text{$s$ on $l$}) \rightarrow v}
  f'(s) \right) \delta_v(t) + \mathrm{propagator\;contribs.}
\end{equation}
(the Dirac distribution $\delta_v(t)$ on the graph is defined as
\begin{equation*}
\int_{\mathfrak g(\tau)} \dif{t} \delta_v(t)\; g(t) = g(v)
\end{equation*}
for a continuous function $g : \mathfrak g(\tau) \rightarrow \mathbb R^d$).
While the first derivative of a function on the graph demands an
orientation of the links, the second derivative is well-defined
without this concept. The rule (\ref{eqn:Laplacegraph}) applies also
at vertices with only \textit{one} internal propagator attached (such
vertices have two or more external propagators attached).

Usually, we need a domain which makes the graph Laplacian
a self-adjoint operator. The treatment of the graph Laplacian is not much
different from the well-known treatment of the one-dimensional Laplacian $\triangle$
on the unit circle $S^1$, since the graphs we are considering are,
with exception of the vertices, compact one-dimensional manifolds.
A self-adjoint domain $\mathcal D(\triangle)$ can be constructed by closing
the subspace of continuous functions with respect to the finite Sobolev norm
$\| f \|_{H^2}^2 = \int_{\mathfrak g(\tau)} \dif{t}
[f^2 + (\partial_t f)^2 + (\triangle f)^2]$; on this domain, the graph Laplacian
is symmetric by integration by parts (the marked difference to the $S^1$-case
is the use of the graph derivative $\partial_t$ in this Sobolev norm).
Since the domain is maximal, this is also the domain of self-adjointness.

Let $t_j \in \mathfrak g(\tau)$ be the point on the graph where the external
momentum $k_j$ enters the graph; then in the world graph path integral,
we have to include a factor $e^{-i k_j \cdot x(t_j)}$.
For a concise notation, the external momentum ``density''
can be modelled by a generalised function
\begin{equation}
\label{def:pbracket}
k\{\tau\} : t \mapsto \sum_j k_j \delta_{t_j}(t).
\end{equation}
The argument $\{\tau\}$ indicates that the positions of the operator
insertions are parametrised by the moduli. Thus we have
\begin{equation}
\sum_j k_j\cdot x(t_j)
= \int_{\mathfrak g(\tau)} \dif{t} x(t) \cdot k_j \delta_{t_j}(t)
= \langle x, k\{\tau\} \rangle.
\end{equation}
The contribution of all operator insertions in the path integral is
then given by a factor $e^{-i\langle x, k\{\tau\} \rangle}$.
We will start with the simpler situation where all propagators have
the same mass $m(t) \equiv m$; then the mass term will contribute a factor
$\exp \left( - m^2 |\mathfrak g(\tau)| \right)$, where
$|\mathfrak g(\tau)|$ is the total length of the graph.

\begin{thm}
\label{thm:worldgraph}
Let $\mathscr G$ be a compact Feynman graph. Let $k_j$, $j=1 \dots n$
be a collection of external momenta. The amplitude corresponding
to the sum of all graphs in the equivalence class
$\mathfrak g = [\mathscr G]$
is given by the formal world graph path integral
\begin{multline}
\label{eqn:thm:world graph}
G_{[\mathscr G]}(k_1, \dots, k_n)
=  (2\pi)^{\frac{nd}{2}}
   \Big( \prod_{\mathrm{vertices}\; v} - c_v \Big)
   \int_{\Mod \mathfrak g} \frac{\dif[\dim \mathfrak g]{\tau}}
                              {\Sym(\mathfrak g(\tau))}\\
   e^{- m^2 |\mathfrak g(\tau)|} Z_0(\mathfrak g(\tau))^{-1}
   \int_{C(\mathfrak g(\tau))} \mathscr D(x)
   \exp - \left( - \frac{1}{4} \langle x, \triangle x \rangle 
    + i \langle x , k\{\tau\} \rangle \right),
\end{multline}
where $Z_0(\mathfrak g(\tau))$ is a (formal) normalisation
depending on the moduli.
\end{thm}
The domain $C(\mathfrak g(\tau))$ is a reminder that
the paths are supposed to be continuous at the vertices.
We stress once more that the positions $t_j$ of the operator
insertions are part of the moduli.
The normalisation $Z_0(\mathfrak g(\tau))$ will be determined below.

There is a subtlety concerning
the symmetry factors in this formula: When the symmetry factors
are taken over from the usual perturbative expansion, they may vary in general
as the positions of the external insertions are varied. On the other hand,
the symmetry factor could be determined with \textit{all external insertions
removed}. Since the symmetry factor is the size of the automorphism group
of the graph, we would expect it to increase generally (since without the external
insertions, there are less distinguishable features on the graph).
When the external insertions are now again included, we have to
take into account all possible insertion positions, and parametrise them
by additional moduli. The point is that features of the graph which are
\textit{indistinguishable} from the point of view of the graph automorphism group
are very well distinguishable from the point of view of the moduli.
This causes an extra multiplicity which exactly cancels
the surplus symmetry factors of the underlying graph without external insertions
(see figure \ref{fig:symmetry} for an illustration).

\begin{figure}[htbp]
\begin{center}

{
\begingroup\makeatletter
\gdef\SetFigFont#1#2{}%
\endgroup%
}

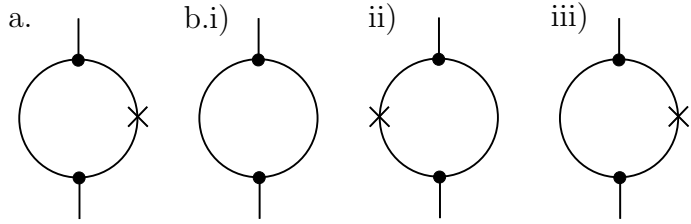
\caption{Example for symmetry factor counting. a. Subgraph from a bigger graph.
The symmetry factor of this loop is 1, since there is an external insertion (X) which
allows to distinguish both loop handles.
b.i) The subgraph without the external insertion is divided by a symmetry factor 2, since
both loop handles are indistinguishable.
ii) + iii) The modulus determining the position of the 
external insertion may now place it on either handle
of the loop. Since the resulting graphs
are topologically equivalent, they add identically and just cancel the factor 1/2.
}
\label{fig:symmetry}
\end{center}
\end{figure}

\section{Interpretation as an effective theory of point particles
  carrying vector charges on the graph}

In order to prove the theorem, we will reformulate the world graph
path integral in such a manner that it will be seen to be
equivalent to the partition function of a classical system of charged
particles moving on the graph. This result is closely related to
the results of Schmidt and Schubert \cite{Schmidt1994}, although we choose a
different language (see also \cite{Lam1993, Feng1994}). It is reminiscent
of the Born-Oppenheimer approximation to the hydrogen molecule, where
after determining the effective potential for the nuclei
mediated by the electrons, one analyses the motion of
the nuclei in this effective potential (integration over moduli
space). This is an interesting result in itself, but it will also
aid the proof.

Consider the Gaussian integral in (\ref{eqn:thm:world graph}).
The integration of the global translation degree of freedom yields the
usual factor $(2\pi)^d\delta^{(d)}(\sum_j k_j)$, so
\begin{multline*}
\int_{C(\mathfrak g(\tau))} \mathscr D(x)
  \exp - \left( - \frac{1}{4} \langle x, \triangle x \rangle 
  + i \langle x , k\{\tau\} \rangle \right) \\
\sim (2\pi)^d \delta^{(d)} \Big( \sum_j k_j \Big)
  \exp \Big \langle k\{\tau\}, \triangle^{-1} k\{\tau\} \Big \rangle.
\end{multline*}
The exponent will be very important later on. It defines an
``interaction potential'' $V_\text{eff}(t_1, \dots, t_n)$ of the
operator insertions by
\begin{align*}
V_\mathrm{eff}(t_1, \dots, t_n)
=& -\Big \langle k\{\tau\}, \triangle^{-1} k\{\tau\} \Big \rangle
\hspace{1cm}\Big(\sum_j k_j=0 \Big).
\end{align*}
If we collect all normalisations into
\begin{equation}
\label{eqn:Zeff}
Z_\mathrm{eff}^{-1}(\mathfrak g(\tau))
= Z_0(\mathfrak g(\tau))^{-1} \; \Big\| \frac{-\triangle}{4\pi} \Big\|_+^{-d/2}
\end{equation}
(read further), then
\begin{multline*}
Z_0(\mathfrak g(\tau))^{-1}
   \int_{C(\mathfrak g(\tau))} \mathscr D(x)
   \exp - \left( - \frac{1}{4} \langle x, \triangle x \rangle 
    + i \langle x , k\{\tau\} \rangle \right) \\
= (2\pi)^d \delta^{(d)} \Big( \sum_j k_j \Big)
  Z_\mathrm{eff}^{-1}(\mathfrak g(\tau))\, e^{- V_\mathrm{eff}(t_1, \dots, t_n)}.
\end{multline*}
This formula deserves a few comments on
the graph Laplacian $\triangle$. It is easy to see that
$\triangle \; 1=0$, so the kernel of $\triangle$ is nonempty and
$\triangle^{-1}$ is not uniquely defined in the first place. We
therefore declare that we wish to study the $\triangle^{-1}$ obeying
\begin{equation}
\langle 1, \triangle^{-1} f \rangle = 0,
\hspace{1cm}f \in \Dom(\triangle^{-1})
\end{equation}
(this scalar product is vector-valued, as the left hand side is a
scalar and the right hand side is a vector).
$-\triangle$ is a positive operator and its kernel consists of the
constant functions. It has a pure point spectrum, as
$\mathfrak g(\tau)$ is compact. 
In the determinant (\ref{eqn:Zeff}), we should therefore ignore the 0 eigenvalue
on the right-hand side as it has been already taken care of in the
explicit inclusion of momentum conservation. This is indicated by the
symbol $\|.\|_+$.

On the other hand, $\langle 1, \triangle g \rangle = 0$; so
$\triangle^{-1} f$ is only defined for $f$ with
$\langle 1, f \rangle = 0$. But
\begin{equation*}
\Big \langle 1, k\{\tau\} \Big \rangle = \sum_j k_j = 0
\end{equation*}
precisely due to momentum conservation.

Na\"ively, we would interpret the effective pair potential between two
insertions $k_i \delta_{t_i}$ and $k_j \delta_{t_j}$ as
\begin{equation*}
-2 \langle k_i \delta_{t_i}, \triangle^{-1} k_j \delta_{t_j} \rangle,
\end{equation*}
however, as it stands, $\delta_{t_j} \not\in \Dom(\triangle^{-1})$
because $\langle 1, \delta_{t_j} \rangle = 1$. There is a canonical
solution to the problem.


{
While $\delta_{t_i}$ is not in the domain of $\triangle^{-1}$,
the difference $\delta_{ij} \equiv \delta_{t_i} - \delta_{t_j}$
certainly is. We may take advantage of this by employing
repeatedly the momentum conservation condition $\sum_j k_j = 0$ (in
the first and third equality) as
\begin{multline*}
\Big\langle \sum_i k_i \delta_{t_i},\,
   \triangle^{-1} \sum_j k_j \delta_{t_j} \Big\rangle
= \Big\langle \sum_i k_i \delta_{i1},\,
   \triangle^{-1} \sum_j k_j \delta_{j1} \Big\rangle
= \sum_{i,j} (k_i \cdot k_j) \Big\langle \delta_{i1},\,
   \triangle^{-1} \delta_{j1} \Big\rangle\\
= \frac{1}{2} \sum_{i,j} (k_i \cdot k_j) \Big\langle \delta_{i1} - \delta_{j1},\,
   \triangle^{-1} (\delta_{j1} - \delta_{i1}) \Big\rangle
= - \sum_{i<j} (k_i \cdot k_j) \langle \delta_{ij},\,
   \triangle^{-1} \delta_{ij} \rangle.
\end{multline*}
Defining the pair potential
\begin{equation}
\label{def:homosapiens}
\varphi(t, t')
= \langle \delta_t - \delta_{t'},\,
          \triangle^{-1} (\delta_t - \delta_{t'}) \rangle,
\end{equation}
the total effective potential is
\begin{equation}
V_\mathrm{eff}(t_1, \dots, t_n)
= \sum_{i<j} (k_i \cdot k_j) \varphi(t_i, t_j).
\end{equation}
While we have used total momentum conservation, we find a pair
potential which is nevertheless independent of the positions of the
other charges on the graph. It is a continuous function on
$\mathfrak g(\tau) \times \mathfrak g(\tau)$ and bounded (so it is weak).
By definition, $\varphi(t,t) = 0$. As $-\triangle^{-1}$ is a
positive operator, in general $\varphi(t, t') \leq 0$. This
implies that the interaction between parallel vector
charges is a repulsive one, since
$(k \cdot k) \varphi(t,t) \geq (k \cdot k) \varphi(t, t')$ for
all positions $t,t' \in \mathfrak g(\tau)$ and charges $k \in \mathbb R^d$:
Spatial separation of the charges is energetically favoured.
} 

\comment{   
The next issue is to obtain clarity about the normalisation
$Z_\mathrm{eff}^{-1}$. Regarding (\ref{eqn:Zeff}), we can see that
the factor $Z_0(\mathscr G)$ cancels precisely
the fluctuation part \textit{between} the
vertices, ie with the coordinates of the vertices fixed;
what is left is the quotient of the fluctuations which
are \textit{linear} between the vertices, and a prefactor for each
propagator ensuring the normalisation to unit integral:
\begin{align}
\label{eqn:Zeff:1}
Z_\mathrm{eff}^{-1}
= \left( \prod_{\mathrm{vertices}\;v} \int \dif[d]{x_v} \right) \delta^{(d)}(x_1)
  \prod_{\mathrm{propagators}\;j} \left( \frac{1}{4\pi \tau_j} \right)^{d/2}
  \mathscr K_{\tau_j}^{1/2}(x_{1(j)}-x_{2(j)}),
\end{align}
where $x_{1(j)}$ and $x_{2(j)}$ are the endpoints of the propagators. As an
example, all three diagrams in fig. \ref{fig:slide} lead to the normalisation
\begin{align*}
Z_\mathrm{eff}^{-1}
=& \int \dif[d]{x}
  \left( \frac{1}{4\pi \tau_3} \right)^{d/2} \mathscr K_{\tau_1}^{1/2}(x)
  \left( \frac{1}{4\pi \tau_2} \right)^{d/2} \mathscr K_{\tau_2}^{1/2}(x)
  \left( \frac{1}{4\pi \tau_1} \right)^{d/2} \mathscr K_{\tau_3}^{1/2}(x)\\
=&  \left( \frac{1}{4^3\pi^3 \tau_1 \tau_2 \tau_3} \right)^{d/2}
  \int \dif[d]{x} \exp - \left( \frac{1}{\tau_1} + \frac{1}{\tau_2}
                             + \frac{1}{\tau_3}\right)
  \frac{x^2}{4} \\
=&  \left( \frac{1}{16 \pi^2
           (\tau_1 \tau_2 + \tau_1 \tau_3 + \tau_2 \tau_3)} \right)^{d/2}.
\end{align*}
} 

Putting everything together, the total amplitude
(\ref{eqn:thm:world graph}) resulting from the
equivalence class $\mathfrak g = [\mathscr G]$ can be written as
\begin{multline}
\label{eqn:G:1}
G_{[\mathscr G]}(k_1, \dots, k_n)
=  (2\pi)^{\frac{nd}{2}}\,
   (2\pi)^d \delta^{(d)} \Big( \sum_j k_j \Big)
   \Big( \prod_{\mathrm{vertices}\; v} - c_v \Big)\\
   \int_{\Mod \mathfrak g} \frac{Z_\mathrm{eff}^{-1}(\mathfrak g(\tau))\,
                               \dif[\dim \mathfrak g]{\tau}}
                              {\Sym(\mathfrak g(\tau))}
   \exp \left( - \sum_{\text{vertices $i<j$}} (k_i \cdot k_j) \varphi(t_i, t_j)
   - m^2 |\mathfrak g(\tau)| \right)
\end{multline}
(note that the interaction has been written here as a true pair
potential, ie the sum extends over each unordered pair $\{i, j\}$ only once).
The (vector valued) potential generated by all charges on the graph is
\begin{equation}
\label{def:potential:varphi}
U_\mathrm{tot}(t) = \sum_{\text{vertices $j$}} k_j \varphi(t, t_j),
\hspace{2cm}t \in \mathfrak g(\tau),
\end{equation}
and we have
\begin{equation}
\label{def:Epot:varphi}
\sum_{\text{vertices $i<j$}} (k_i \cdot k_j) \varphi(t_i, t_j)
= \frac{1}{2} \sum_{\text{vertices $i$}} k_i \cdot U_\mathrm{tot}(t_i).
\end{equation}
The potential on the graph fulfills
\begin{equation}
\label{eqn:deq:Pot}
\triangle U_\mathrm{tot}(t) = -2 k\{\tau\},
\end{equation}
so in between the insertions, $U_\mathrm{tot}(t)$ is a linear function (with
vanishing second derivative).

\subsection{Proof of theorem \ref{thm:worldgraph}}

The proof starts from the direct Schwinger representation
(\ref{eqn:G:massless}) valid also for massless propagators.
We introduce a matrix notation for the exponent.
Define the symmetric covariance matrix $C \in M_V(\mathbb R)$ as
follows: if $v \not= w$, then
\begin{equation}
\label{def:Cvw}
C_{vw} = - \frac{1}{2} \sum_{j(v \leftrightarrow w)} \tau_j^{-1}
\hspace{1cm}(v \not= w).
\end{equation}
The sum extends over all propagators $j$ connecting $v$ and $w$
directly; if there is no propagator connecting $v$ and $w$ directly
then the matrix element $C_{vw} = 0$.
The diagonal elements are then chosen in such a way that
the sum of each row/column equals zero. The matrix $C$ is a linear
combination of elementary matrices of the form
\begin{equation*}
\Em^{vw}_{ij}
= \delta_i^v \delta_j^v + \delta_i^w \delta_j^w
  - \delta_i^v \delta_j^w - \delta_i^w \delta_j^v.
\end{equation*}
These matrices generate exactly the squares of the coordinate differences
\begin{equation*}
x^T \Em^{vw} x = (x_v - x_w)^2,
\end{equation*}
with the silent understanding that the entries $x_v \in \mathbb R^d$
of the vector $x$ are themselves coordinate vectors. In terms
of these building blocks,
\begin{equation}
\label{eqn:def:C}
C = \frac{1}{2} \sum_{v < w} \Big( \sum_{j(v \leftrightarrow w)}
  \tau_j^{-1} \Big) \Em^{vw}.
\end{equation}
Using the matrix $C$, we may write the amplitude (\ref{eqn:G:massless}) as
\begin{multline*}
G_{\mathscr G}(k_1, \dots, k_n)
=  \frac{(2\pi)^{\frac{nd}{2}}}{\Sym(\mathscr G)}\,
   \Big( \prod_{\mathrm{vertices}\; v} - c_v
   \int \dif[d]{x_v} \Big)\\
   \Big( \prod_{\mathrm{propagators}\; j}
   \int_0^\infty \frac{\dif{\tau_j}}{(4\pi\tau_j)^{d/2}} \Big)
   \exp \Big( - \frac{1}{2} x^T C x - i k^T x - m^2 \sum_j \tau_j \Big).
\end{multline*}
We can say a few things about the spectrum of $C$.
Let $e=\frac{1}{\sqrt V}(1, 1, \dots, 1)^T$. By construction, $C e = 0$.
Because we consider connected graphs, the kernel of $C$ contains only
multiples of $e$. Furthermore, because $x^T C x$ is a sum of
squares and never vanishes except when all $x_v$ coincide, $C$ is
strictly positive with the exception of the eigenspace generated by
$e$.

For the purpose of integrating $\dif[d]{x_v}$, we have to
invert the singular matrix $C$. One can easily see that the 0
eigenvalue again enforces momentum conservation. The generalised
inverse
\begin{equation*}
C^\text{inv} = \lim_{c \rightarrow \infty} (C + cee^T)^{-1}.
\end{equation*}
always exists, because the eigenvector $e$ decouples,
and $C$ is strictly positive elsewhere.
When we insert the momentum-conserving $\delta$-distribution in the
end, we have to include a factor $V^{d/2}$ because of the
normalisation of the eigenvector $e$. The amplitude is thus
\begin{multline}
\label{eqn:G:Zeff}
G_{\mathscr G}(k_1, \dots, k_n)
= \frac{(2\pi)^{\frac{nd}{2}}}{\Sym(\mathscr G)}\,
(2\pi)^d \delta^{(d)}\Big( \sum_j k_j \Big)
\Big( \prod_{\text{vertices $v$}} - c_v\Big) \\
\Big( \prod_{\text{propagators $j$}} \int_0^\infty
   \frac{\dif{\tau_j}}{(4 \pi \tau_j)^{d/2}}\Big)
   V^{d/2} \Big\| \frac{C}{2 \pi} \Big\|_+^{-d/2}
\exp \Big( - \frac{1}{2} k^T C^\text{inv} k
     - m^2 \sum_j \tau_j \Big).
\end{multline}
The non-singular part of the determinant can be obtained by
\begin{equation*}
\| C \|_+ = \det (C + ee^T).
\end{equation*}
Noting that the ``ordinary'' Schwinger parameters $\tau_j$ which we
have used are in one-to-one correspondence to the continuous moduli of 
the equivalence class scheme B, we can easily obtain a sum over the
classes of scheme A or scheme C by summing over the necessary
orderings of the operator insertions, and over the discrete moduli
defining the permutations of the insertions, and possibly over graph
topologies (this is possible since the equivalence classes B are
subclasses of A and C). Hence, equality with theorem
\ref{thm:worldgraph} is established if we can show that the
exponential coincides with the one in (\ref{eqn:G:1}), ie if
\begin{equation}
\label{eqn:Cvarphi}
C^\text{inv}_{ij} = \varphi(t_i, t_j).
\end{equation}
$C^\text{inv}$ should be the matrix analog to the pair potential
$\varphi$. Rather than proving this formula directly, we
compute
\begin{multline*}
\sum_{jl} C_{ij} \varphi(t_j, t_l) k_l
= \sum_j C_{ij} U_\mathrm{tot}(t_j) \\
= \frac{1}{2} \sum_{\substack{\text{propagators}\;l\\
                          \text{ending at vertex $2(l) = i$}}}
  \frac{U_\mathrm{tot}(t_i) - U_\mathrm{tot}(t_{1(l)})}{\tau_l}
= \frac{1}{2} \sum_{\substack{\text{prop.s $l$}\\
                              \text{with $2(l) = i$}}}
   U_\mathrm{tot}'(t) \Big|_{\text{$t$ on $l$}}
=  k_i
\end{multline*}
by the fact that the potential is linear along the propagators,
and equation (\ref{eqn:deq:Pot}). This proves (\ref{eqn:Cvarphi}).
By comparison, we find that the normalisation constant must be given by
\begin{equation}
\label{eqn:Zeff:C}
Z_\text{eff}^{-1}(\mathfrak g(\tau))
= V^{d/2} \Big\| \frac{C}{2 \pi} \Big\|_+^{-d/2}
  \prod_{\text{propagators}\;j} \frac{1}{(4 \pi \tau_j)^{d/2}}.
\end{equation}
This concludes the proof.

As a side-effect, we have found a closed formula for the normalisation
$Z_\text{eff}^{-1}(\mathfrak g(\tau))$. In many cases, the following
lemma states a more useful form:

\begin{lem}
\label{lem:Zeff:1}
The measure $Z_\text{eff}^{-1}(\mathfrak g(\tau)) \dif[\dim \mathfrak g]\tau$
on moduli space $\Mod \mathfrak g$ is given by
\begin{align*}
Z_\mathrm{eff}^{-1}(\mathfrak g(\tau))
= \Big( \prod_{\mathrm{vertices}\;v} \int \dif[d]{x_v} \Big) \delta^{(d)}(x_1)
  \prod_{\mathrm{propagators}\;j} \left( \frac{1}{4\pi \tau_j} \right)^{d/2}
  e^{-\frac{1}{4 \tau_j} (x_{1(j)}-x_{2(j)})^2},
\end{align*}
where $x_{1(j)}$ and $x_{2(j)}$ are the endpoints of the propagators
and $x_1$ is an arbitrary vertex on the graph.
\end{lem}
\begin{proof}
Equate (\ref{eqn:G:Zeff}) and (\ref{eqn:G:massless}). Pick an
arbitrary vertex $x_1$. Integrate $\int \dif[d]{k_1}$ and
put all other external momenta $k_j$ to zero, $j=2\dots V$.
The proposed expression is obtained by substituting equation
(\ref{eqn:Zeff:C}).
\end{proof}

\subsection{Coincidence of vertices: Cell structure of moduli space. Renormalisation.}
\label{sec:Richard}

If the modulus $\tau_j$ associated to any one
propagator in a graph $\mathfrak g(\tau)$ shrinks to 0 (see
fig. \ref{fig:fuse}), the two adjacent vertices concur
in the limit.
\begin{figure}
\centering \includegraphics[width=11cm]{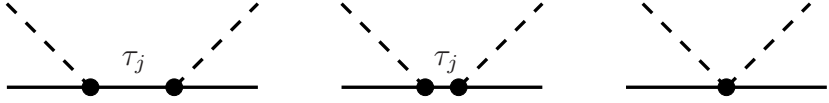}
\put(-150,12){$\tau_j$}
\put(-268,12){$\tau_j$}
\caption{(left) Section of metric graph $\mathfrak g(\tau)$ containing
  a propagator  with modulus $\tau_j$. The amplitude contains the
  product $c_3^2$ of coupling constants.
  (middle) While $\tau_j \rightarrow 0$, the insertions are approaching
  each other. Still, the amplitude is proportional $c_3^2$.
  (right) After the insertions fuse, the new metric graph
  $\mathfrak g'(\tau')$ is proportional $c_4$.}
\label{fig:fuse}
\end{figure}
It is important to realise that there is a fundamental difference
between the limit $\tau_j \rightarrow 0$ of the graph $\mathfrak g(\tau)$,
and the graph $\mathfrak g'(\tau')$ containing a single vertex,
obtained through fusion of the pair of adjacent vertices
($\tau'$ are ``reduced'' moduli, removing $\tau_j$ from
the moduli $\tau$).
While the modulus $\tau_j$ is a continuous parameter,
a coordinate parametrising the integrand of the moduli space integration, the
hypersurface $\tau_j=0$ lies on the boundary of
$\Mod \mathfrak g$ and therefore has measure zero; so it does not
contribute (for divergences, so below).

In contrast, the class $\mathfrak g'$ - viewed as an
independent contribution to the total correlation -
has a nonvanishing measure in general. That $\mathfrak g'(\tau')$ and
$\left. \mathfrak g(\tau) \right|_{\tau_j = 0}$
are truly different can also be seen from that fact that the valency of
the fused vertex and therefore the prefactor assembled from the
product of coupling constants is different.
In technical terms, $\mathfrak g'$ specifies a different
cell of moduli space.

It is interesting to study the global structure of the
complete moduli space, and examine the relations between different
cells. Following the literature, we claim that
$\Mod \mathfrak g' \subset \partial \Mod \mathfrak g$, ie the
cell resulting if one\ed (or several) moduli $\tau_j \rightarrow 0$ makes up part of
the boundary of the original cell \cite{Gopakumar:2005fx}.
For the dimensions of the adjacent cells $\mathfrak g$ and
$\mathfrak g'$, obviously $\dim \mathfrak g' < \dim \mathfrak g$.
In this way, the moduli space has a
natural complex structure (in the topological sense).

Another kind of boundary is reached in the limit
$\tau_j \rightarrow \infty$:
Momentum transfer through the propagator is increasingly suppressed;
this cell boundary is made up of the
graph where the propagator is missing out altogether.
Note that the limiting graph still contains \textit{two} propagators on both
sides, as well as additional 2-valent vertices (mass terms).
It may happen that the graph falls
apart into two components in this limit; a systematic treatment
therefore has to include disconnected graphs.

\commentpaper
{
\begin{figure}
\centering \includegraphics[width=8cm]{threefour.eps}
\caption{The graphs (b) in the limit $\tau \rightarrow 0$ degenerate
  into (a). Under $\tau\rightarrow \infty$, they will approximate (c).}
\label{fig:threefour}
\end{figure}

For concreteness, consider the graph in fig. \ref{fig:threefour} (a).
The 4-valent vertex is the $\tau
\rightarrow 0$~limit of either one of the three graphs (b) in the $s$-,
$t$- and $u$-channel. If we take the limit $\tau \rightarrow
\infty$ on the other hand, then as limiting graphs we get (c). As mentioned, the latter
case requires that we extend the world graph formalism to graphs with
several connection components. There is a subtlety concerning
(c): The graphs still contain \textit{two} propagators on both
unconnected components; we should have correspondingly a 2-valent
vertex in the theory (a mass term).

At the same time, the moduli space cells for the graphs in (b) are
themselves on the boundary of a higher-dimensional moduli space cell,
e.g. one having a second propagator parallel to the propagator connecting the
two vertices.
}

This opens the door for speculations whether there occurs an ultimate
simplification in the amplitude when we extend the sum over all different
cells of moduli space (prior to integration of the moduli).
There is one context where this is indeed required, namely in
renormalisation. Rewriting the na\"ive Schwinger
parametrisation of the propagator as
\begin{equation}
\label{eqn:UVdivergent}
\lim_{\varepsilon \rightarrow 0+} \int_\varepsilon^\infty
\dif{\tau_j} e^{-\tau_j(q^2 + m^2)}
= \lim_{\varepsilon \rightarrow 0+}
  \frac{e^{-\varepsilon(q^2 + m^2)}}{q^2 + m^2},
\end{equation}
in the limit $\varepsilon \rightarrow 0$, the
suppression of high $q^2$ momentum contributions due to the
regularising exponent vanishes and a UV divergence
is a possible consequence. We make the following
\begin{ass*}
Cancellations due to renormalisation are \textit{local} in moduli space.
\end{ass*}
It is immediate that the necessary (formally infinite)
counterterms must come from the neighbouring moduli space cell
reached in the limit $\tau_j \rightarrow 0$.

For IR divergences, rewrite the Schwinger parametrisation of the massless
propagator as
\begin{equation}
\label{eqn:IRdivergent}
\lim_{\varepsilon \rightarrow \infty} \int_0^\varepsilon
\dif{\tau_j} e^{-\tau_j q^2}
= \lim_{\varepsilon \rightarrow \infty}
  \frac{1 - e^{-\varepsilon q^2 }}{q^2};
\end{equation}
it is the limit $\tau_j \rightarrow \infty$ which is responsible
for a possible IR divergence at $q^2 \rightarrow 0$. Correspondingly,
the IR ``counterterms'' are to be obtained from the diagram
without the propagator in question. It has a regularisation
mass term coming from the unobservable background modes coupling to the
fields.

Viewing renormalisation in the geometrical moduli space picture, it
is obvious that the cancellation of divergences is independent of the
particular moduli space parametrisation. A renormalisation example
will be given in section \ref{sec:spider}.

As an intriguing possibility, it is imaginable that the moduli space
integrations can be formulated as integrals of a total divergence; in that case,
by Stokes's Theorem, they might be reduced to an integral over boundary terms \textit{only},
even if we cannot expect the boundary terms to have a direct graphical interpretation.
By iteration of this procedure, amplitudes could be computed as integrals
over the lowest dimensional boundary cells of moduli space (moduli space
``effective vertices'')
\footnote{Such iteration might require very particular relations between the
coupling constants and masses of the system, as in non-Abelian gauge theories.
The case of gauge theories is indeed special: Here,
the couplings are fixed ab initio by the requirement of gauge
invariance.}.


\subsection{Example: Tree diagrams}

The most elementary examples of connected Feynman graphs are tree
diagrams. The vertices (insertions) at the ``endpoints'' of the tree
are connected to the rest of the diagram only by a single internal
line (so they have internal valency 1), and they are linked to at
least two external lines (external valency $\geq 2$). ``Internal''
branching vertices (internal valency $\geq 3$)
may or may not have a non-zero external valency as
well. Finally, there are vertices with internal valency two and
external valency $\geq 1$.

To cover the whole equivalence class $\mathfrak g$ of the tree, we
have to sum over all topologically inequivalent groupwise permutations
of the external legs whenever several external legs are located on
insertions with the same valences; insertions with internal valency 2
are allowed to slide all over the tree. As the external legs are all
distinguishable, there are no symmetries of the graph, so the symmetry
factor equals 1.

Let there be $n$ insertions with internal valency 2. We denote their
positions on the metric graph by $t_j \in \mathfrak g(\tau)$,
and the total external momentum entering at insertion $j$ by
$k_j$. Removing these,
let there be $V$ vertices left with internal valency other than 2,
connected by $P$ lines. Denote the lengths of these lines by the
moduli $\tau_j \in \mathbb R_+$. Clearly, the dimension of the graph
(the number of continuous moduli) is $\dim \mathfrak g = n + P$.

Given any two insertions $t_1, t_2 \in \mathfrak g(\tau)$, we want to find
the interaction potential for the momenta entering at these insertions
(the charges). We have to determine
\begin{equation*}
f(t) = \triangle^{-1} (\delta_{t_1} - \delta_{t_2})(t),
\hspace{1cm}t \in \mathfrak g(\tau).
\end{equation*}
The function $f(t)$ is easily characterized: Let
$\mathcal T(t_1, t_2)$ denote the path from $t_1$ to $t_2$, and
$|\mathcal T(t_1, t_2)|$ be its length (in terms of the moduli $\tau_j$).
$f(t)$ is continuous by definition; it is piecewise constant on all
segments of the graph except $\mathcal T(t_1, t_2)$;
and on $\mathcal T(t_1, t_2)$ it increases linearly with the distance
from $t_1$. The absolute value of $f(t)$ is unimportant,
so an arbitrary constant may be added. It follows that
\begin{equation*}
\varphi(t_1, t_2)
= \langle \delta_{t_1} - \delta_{t_2},
  \triangle^{-1} (\delta_{t_1} - \delta_{t_2}) \rangle
= f(t_1) - f(t_2)
= - |\mathcal T(t_1, t_2)|.
\end{equation*}
For the effective normalisation, one can see that
$Z_\mathrm{eff}^{-1}(\mathfrak g(\tau)) = 1$ by starting to integrate
the formula given in lemma \ref{lem:Zeff:1} at the
vertices forming the tips of the tree, and working down towards
$v_1$ which is an arbitrary vertex on the tree. The integrals then
always cancel exactly the prefactor.

So for a tree diagram, we find the amplitude
\begin{multline*}
G_\text{tree}(k_1, \dots, k_n)
= (2\pi)^{\frac{nd}{2}} (2\pi)^d \delta^{(d)} \Big( \sum_j k_j \Big)
\Big( \prod_{\mathrm{vertices}\; v} - c_v \Big) \\
\sum_{\text{perm.s of $k_j$}}
\int_0^\infty \dif[P]{\tau}
\int_{\mathfrak g(\tau)} \dif[n]{t}
\exp \Big(\sum_{v<w} (k_v \cdot k_w) |\mathcal T(t_v, t_w)|
          - m^2 \sum_j \tau_j \Big).
\end{multline*}
It is an amusing exercise to see how this expression resolves into the
correct amplitude upon integrating the moduli.

\subsection{Example: The one-loop ``spider'' diagram}
\label{sec:spider}

The simplest diagram containing a loop integration is the one-loop amplitude with an
arbitrary number $n$ of external legs \footnote{A yet unknown species
  of spiders.} coupled by three-valent vertices $c_3$ directly to the
loop. The class $\mathfrak g$ contains all permutations of the
order of external legs; there is a cyclic symmetry.

Let the incoming momenta be $k_1, \dots, k_n$. Let $t_j$ be the
coordinate which describes the
position of the $j$-th vertex entering the loop with respect to some
fixed parametrisation of the loop. Furthermore, let $\tau$ be the
total length of the loop. The true moduli are given by the distances
\textit{between} the insertions; so if we integrate the coordinates
$t_j$ instead, we have to include a factor $\frac{1}{\tau}$ to
account for the arbitrary choice of an origin of the parametrisation.
For the normalisation, one finds
\begin{align*}
Z_\mathrm{eff}^{-1}(\mathfrak g(\tau))
=& \left( \frac{1}{4 \pi \tau} \right)^{d/2}
\end{align*}
by convolution of the Wiener kernels of lemma \ref{lem:Zeff:1}.
Assume that $t_i < t_j$. Then,
\begin{equation*}
\triangle^{-1} (\delta_{t_i} - \delta_{t_j})(t)
= \begin{cases}
\left( \frac{t_j - \tau + t_i}{2} - t \right) \frac{t_j - t_i}{\tau}
& \text{if $t < t_i$},\\
\left( t - \frac{t_i + t_j}{2}\right) \frac{\tau - t_j + t_i}{\tau}
& \text{if $t_i \leq t < t_j$},\\
\left( \frac{t_j + \tau + t_i}{2} - t \right) \frac{t_j - t_i}{\tau}
& \text{if $t_j \leq t$}
\end{cases}
\end{equation*}
at the point with coordinate $t$ with respect to the parametrisation.
This is easily checked by applying $\triangle$. Thus,
\begin{align*}
\varphi(t_i, t_j)
=& \langle \delta_{t_i} - \delta_{t_j},
   \triangle^{-1} (\delta_{t_i} - \delta_{t_j}) \rangle
=  - \frac{| t_j - t_i | (\tau - |t_j - t_i| )}{\tau}.
\end{align*}
In the last form, the potential is valid also for $t_i > t_j$.
As the orientation of the loop parametrisation is arbitrary,
we have to include a symmetry factor $\frac{1}{2}$, and we get
a total amplitude
\begin{multline}
\label{eqn:spider}
G_\text{1-loop}(k_1, \dots, k_n)
= (2\pi)^{\frac{nd}{2}} (2\pi)^d \delta^{(d)} \Big( \sum_j k_j \Big)
\left( - c_3 \right)^n
\int_0^\infty \frac{\dif{\tau}}{2\tau} \left( \frac{1}{4\pi \tau} \right)^{d/2}\\
\exp \left( - m^2 \tau \right)
\int_0^\tau \dif[n]t \exp \left( \sum_{i<j} (k_i \cdot k_j)
   \frac{| t_j - t_i | (\tau - |t_j - t_i| )}{\tau} \right).
\end{multline}
Written in this form, it is plausible that the presummation of
sufficiently large equivalence classes of diagrams amounts to a
stringent organisation of the amplitudes.\ed Had we used the
equivalence class scheme B (the world line formalism), the cyclic
order of external legs entering the loop would be unalterable, and we
should sum explicitly over all such orderings.

This provides a simple example illustrating how renormalisation fits into the scheme:
By rescaling the parameters $t_j \to \tau t_j$, the integral takes the form
\begin{multline*}
G_\text{1-loop}(k_1, \dots, k_n)
= (2\pi)^{\frac{nd}{2}} (2\pi)^d \delta^{(d)} \Big( \sum_j k_j \Big)
\left( - c_3 \right)^n
\int_0^\infty \frac{\dif{\tau}}{2\tau} \left( \frac{1}{4\pi \tau} \right)^{d/2} \tau^n\\
\exp \left( - m^2 \tau \right)
\int_0^1 \dif[n]t \exp \left( \tau \sum_{i<j} (k_i \cdot k_j)
   | t_j - t_i | (1 - |t_j - t_i| ) \right).
\end{multline*}
The integrand has a power series expansion in $\tau$ around the origin
\begin{equation}
\label{eqn:Luke}
c_{n - d/2 - 1} \tau^{n - d/2 - 1}
+c_{n - d/2} \tau^{n - d/2}
+c_{n - d/2 + 1} \tau^{n - d/2 + 1} + \dots;
\end{equation}
the integral diverges at $\tau \to 0$ if $n - d/2 \leq 0$.
This is a UV divergence, by the reasoning of section \ref{sec:Richard}.
The graphs containing the counterterms are found in the moduli space cell
$\tau \to 0$. They are just given by formally integrating the divergent
terms of the power series expansion (\ref{eqn:Luke}):
\begin{align*}
G_{n - d/2 - 1}(k_1, \dots, k_n)
=& - (2\pi)^{\frac{nd}{2}} (2\pi)^d \delta^{(d)} \Big( \sum_j k_j \Big)
\left( - c_3 \right)^n
\int_0^\infty \frac{\dif{\tau}}{2\tau} \left( \frac{1}{4\pi \tau} \right)^{d/2}
\tau^n,\\
G_{n - d/2}(k_1, \dots, k_n)
=&- (2\pi)^{\frac{nd}{2}} (2\pi)^d \delta^{(d)} \Big( \sum_j k_j \Big)
\left( - c_3 \right)^n
\int_0^\infty \frac{\dif{\tau}}{2\tau} \left( \frac{1}{4\pi \tau} \right)^{d/2} \tau^{n+1}\\
& \left( - m^2 + (k_i \cdot k_j) \int_0^1 \dif[n]t
         \sum_{i<j} | t_j - t_i | (1 - |t_j - t_i| ) \right),
\end{align*}
etc. Naturally, the counterterms for higher order divergences are of derivative type.
With these subtractions in place, the total amplitude is finite.

\comment
{
The magnitude of this expression lends itself to a crude
approximation: The magnitude of the pair potential is given by
\begin{equation*}
\varphi(t_i, t_j) \geq \varphi(0, \frac{\tau}{2})
= -\frac{\tau}{4}.
\end{equation*}
The interaction potential is therefore comparable to
\begin{multline*}
V_\mathrm{eff}
= \sum_{i<j} (k_i \cdot k_j) \varphi(t_i, t_j)
\approx - \frac{\tau}{4} \sum_{i<j} (k_i \cdot k_j) \\
= - \frac{\tau}{8} \Big( \sum_{i,j} (k_i \cdot k_j)
      - \sum_i k_i^2 \Big)
= \frac{\tau}{8} \sum_i k_i^2.
\end{multline*}
In this approximation, it is independent of the insertion points $t_j$, and
the integral $\int \dif[n]t = \tau^n$ (this contains all orderings
of the external insertions). The correlations are finally
\begin{multline*}
(2\pi)^{\frac{nd}{2}}
(2\pi)^d \delta^{(d)} \Big( \sum_j k_j \Big)
\left( - c_3 \right)^n
\int_0^\infty \frac{\dif{\tau}}{2\tau} \left( \frac{1}{4\pi \tau} \right)^{d/2}
\exp \left( - m^2 \tau \right)
\tau^n \exp - \frac{\tau}{8} \sum_i k_i^2\\
= (2\pi)^{\frac{nd}{2}} \delta^{(d)} \Big( \sum_j k_j \Big)
\left( - c_3 \right)^n \Gamma\left(n - \frac{d}{2} \right)
\frac{\pi^{d/2}}{2}
\Big( m^2 + \frac{1}{8} \sum_i k_i^2 \Big)^{d/2 - n}.
\end{multline*}
For the self-energy ($n=2$), this gives indeed a qualitatively good
result for small external momentum.
}

\subsection{Example: Two-loop self-energy graph}

At two-loop level, the sliding of external legs over internal vertices is
important. We discuss the contribution of the graph given in figure \ref{fig:twoloop}
to the self-energy in scalar $\phi^3$-theory.

\begin{figure}[htbp]
\begin{center}

{
\begingroup\makeatletter
\gdef\SetFigFont#1#2{}%
\endgroup%
}

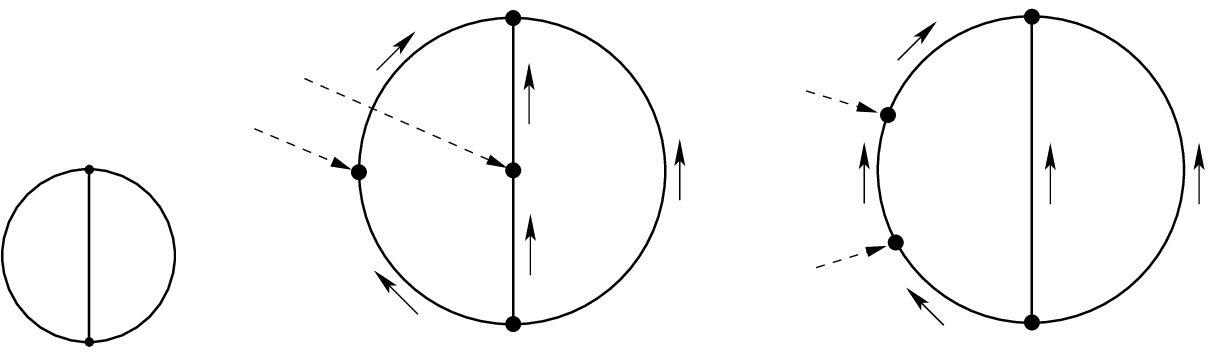
\caption{Notation in the two-loop computation.
$v_1, v_2$: external insertions.
$k_1, k_2$: external momenta.
$v_3, v_4$: internal vertices.
$\tau_1, \tau_2, \tau_3$ (inset): length moduli of branches between $v_1$ and $v_2$.
$t_1, t_1', t_2$: length moduli of segments between vertices / insertions.
$m_1, m_2,\dots$, arrows: gradients of the piecewise linear function
  $\triangle^{-1}_v (\delta_{v_1} - \delta_{v_2})(v)$ along the branches of the graph.
}
\label{fig:twoloop}
\end{center}
\end{figure}

There are two different topologies:
Both external insertions $v_1$ and $v_2$ may
slide along the same internal line between the two vertices (a), or they may be along two
separate lines (b). With the moduli parametrising the metric distances along the
propagators as defined in the figure, the first task is the computation of the
pair potentials $\varphi_{(a)}(v_1, v_2)$ resp. $\varphi_{(b)}(v_1, v_2)$, as
defined in (\ref{def:homosapiens}).

We need to compute the action of the inverse Laplacian
$\triangle^{-1}_v (\delta_{v_1} - \delta_{v_2})(v)$ for points $v \in \mathfrak g(\tau)$ on the graph.
We surely know that this is a continuous function, piecewise linear between vertices.
Denote the respective gradients
$\partial_v [\triangle^{-1}_v (\delta_{v_1} - \delta_{v_2})](v)$ in (a) by
$m_1, m_1', m_2, m_2', m_3$, as in the figure.
These gradients fulfill the equations
\begin{eqnarray*}
m_1' - m_1 &=& 1\\
m_2' - m_2 &=& -1\\
t_1 m_1 + (\tau_1 - t_1) m_1' = t_2 m_2 + (\tau_2 - t_2) m_2' &=& \tau_3 m_3 \\
m_1 + m_2 + m_3 &=& 0
\end{eqnarray*}
(the first pair of equations normalises the $\delta$ source terms; the second
pair says that the potential is continuous at the vertices $v_3$ and $v_4$;
and the last equation declares that $v_4$ is uncharged).
There is a another similar but redundant condition at the vertex $v_3$.
The solution is, with $\Delta = \tau_1 \tau_2 + \tau_1 \tau_3 + \tau_2 \tau_3$,
\begin{align*}
m_1' =& \frac{(\tau_2 + \tau_3) t_1 + \tau_3 t_2}{\Delta}
 & m_1 = m_1' - 1\\
m_2' =& -\frac{(\tau_1 + \tau_3) t_2 + \tau_3 t_1}{\Delta}
 & m_2 = m_2' + 1\\ 
m_3  =& \frac{\tau_1 t_2 - \tau_2 t_1}{\Delta}
\end{align*}
For the pair potential between insertions $v_1$ and $v_2$ in case (a), this implies
\begin{align}
\varphi_{(a)}(v_1, v_2)
=& \langle \delta_{v_1} - \delta_{v_2},
   \triangle^{-1} (\delta_{v_1} - \delta_{v_2}) \rangle
=  t_1 m_1 - t_2 m_2\\
=& - \frac{\tau_1 (\tau_2 - t_2) t_2 + \tau_2 (\tau_1 - t_1) t_1
  + \tau_3(t_1 + t_2)(\tau_1 + \tau_2 - t_1 - t_2)}{\Delta}. \nonumber
\end{align}
A similar computation for case (b) reveals that
\begin{align}
\varphi_{(b)}(v_1, v_2)
=& -\frac{t_1' \left[ (\tau_2 + \tau_3) (\tau_1 - t_1') + \tau_2 \tau_3 \right]}{\Delta}.
\end{align}
This second potential does not depend on the modulus $t_1$.

These potentials are really only pieces or ``branches'' of a \textit{single}
potential function $\varphi(v_1, v_2)$, continuous for all
$v_1 \times v_2 \in \mathfrak{g}(\tau) \times \mathfrak{g}(\tau)$.
They can be connected e.g. when the external insertion $v_2$ crosses through
$v_3$:
\begin{multline*}
\lim_{v_2 \to v_3} \varphi_{(a)}(v_1, v_2)
= \lim_{t_2 \to \tau_2} \varphi_{(a)}(v_1, v_2)
= -\frac{(\tau_1 - t_1) \left[ (\tau_2 + \tau_3) t_1 + \tau_2 \tau_3 \right]}{\Delta}\\
= \lim_{t_1' \to \tau_1 - t_1} \varphi_{(b)}(v_1, v_2)
= \lim_{v_2 \to v_3} \varphi_{(b)}(v_1, v_2).
\end{multline*}
A second way to pass from sheet (a) to sheet (b) is to pass $v_1$ through $v_4$
and use the symmetry of the graph. Again, the potential is continuous.

The computation of the normalisation constant 
$Z_\text{eff}^{-1}$ by equation (\ref{eqn:Zeff:C}) is straightforward.
There are $V=2$ vertices in the naked graph; the matrix $C$ according to
(\ref{eqn:def:C}) is
\begin{equation*}
C = \frac{1}{2} \Big( \frac{1}{\tau_1} + \frac{1}{\tau_2} + \frac{1}{\tau_2} \Big) \cdot
    \begin{pmatrix} 1 & -1 \\ -1 & 1 \end{pmatrix};
\end{equation*}
and thus with
$\| C \|_+ = \det (C + ee^T) = 2( \frac{1}{\tau_1} + \frac{1}{\tau_2} + \frac{1}{\tau_2})$,
the normalisation becomes
\begin{align*}
Z_\mathrm{eff}^{-1}(\mathfrak g(\tau))
=& \left( \frac{1}{32 \pi^2 \Delta} \right)^{d/2}.
\end{align*}
The symmetry factors for both graphs are 2. The total amplitude is therefore
\begin{multline}
G_\text{2-loop}(k_1, k_2)
=  (2\pi)^{2d} \delta^{(d)} \Big( k_1 + k_2 \Big)\;
   \frac{\left( - c_3 \right)^4}{2}
   \int_0^\infty \dif[3]{\tau}
   \left( \frac{1}{32 \pi^2 (\tau_1 \tau_2 + \tau_1 \tau_3 + \tau_2 \tau_3)} \right)^{d/2}\\
   e^{- m^2 (\tau_1 + \tau_2 + \tau_3)}
   \cdot \Bigg\{
   \int_0^{\tau_1} \dif{t_1} \int_0^{\tau_2} \dif{t_2}
   \exp \left( - (k_1 \cdot k_2) \varphi_{(a)}(v_1, v_2) \right)\\
   + \int_0^{\tau_1} \dif{t_1'} (\tau_1-t_1')
   \exp \left( - (k_1 \cdot k_2) \varphi_{(b)}(v_1, v_2) \right)
   \Bigg\}.
\end{multline}
The factor $\tau_1-t_1' = \int_0^{\tau_1-t_1'} \dif{t_1}$ comes from integrating
the irrelevant modulus $t_1$.


\subsection{Particles of different mass} 
\label{sec:mass}

We comment briefly on how to include particles of different
mass. For being explicit, consider the self-energy diagram
fig. \ref{fig:phi3-se}. Assume that the internal
propagators carry different masses $m_1$ and $m_2$, and that the
vertices always couple to one external and one of each particles of
masses $m_1$, $m_2$. The total mass exponential can then be written
\begin{multline*}
e^{-m_1^2 \tau_1 - m_2^2 \tau_2} + e^{-m_2^2 \tau_1 - m_1^2 \tau_2} \\
= \Tr \left\{
\begin{pmatrix} e^{-m_1^2 \tau_1} & 0 \\ 0 & e^{-m_2^2 \tau_1} \end{pmatrix}
\begin{pmatrix} 0 & 1 \\ 1 & 0 \end{pmatrix}
\begin{pmatrix} e^{-m_1^2 \tau_2} & 0 \\ 0 & e^{-m_2^2 \tau_2} \end{pmatrix}
\begin{pmatrix} 0 & 1 \\ 1 & 0 \end{pmatrix} \right\}.
\end{multline*}
The first and third matrix are representing the propagators; the
other two matrices are the vertices ``switching'' between different
masses. This might seem overly formal. However, it is only
consequent. Namely, observe that this is a trace over a path-ordered
exponential
\begin{equation*}
\Tr P_{\mathfrak g(\tau)}\;
\left\{
\begin{pmatrix} 0 & 1 \\ 1 & 0 \end{pmatrix}_{t_1}
      \begin{pmatrix} 0 & 1 \\ 1 & 0 \end{pmatrix}_{t_2}
\exp - \int_0^\tau \dif{t} 
  \begin{pmatrix} m_1^2 & 0 \\ 0 & m_2^2 \end{pmatrix} \right\},
\end{equation*}
where $t_1$ and $t_2$ denote the locations of the operator insertions,
$\tau$ is the total length of the loop, and $P_{\mathfrak g(\tau)}$ is the
path ordering on $\mathfrak g(\tau)$. This construction can be generalised to
more complex graphs (although a matrix notation for the vertices is clearly
not possible).
The resulting mass exponential is then part of the integrand of the
moduli space integral, in the spirit of the world graph formalism.
Similar techniques can be applied if the propagators are gauge bosons
in the adjoint representation of some local gauge group
(for the application of path ordering to inclusion of background
potentials see \cite{Fliegner1994}).


\section{Vector and tensor particles and the dipole method}

So far, we have considered scalar interaction vertices and scalar particles. In
general, vertices might contain derivatives of the adjoining
propagators; when tensor particles are involved, their propagators
will contain supplementary polynomials in their momenta,
and the propagator term $(q_j^2 + m^2)^{-1}$ might be raised to a higher
power. While the latter can be generated by including additional
factors $\tau_j$ in the measure on moduli space (cf. section \ref{sec:conf}),
there remains a prefactor multiplying the integrand of a Feynman amplitude
which is a polynomial in the momenta $q_i$
of the propagators and $k_i$ of the external
legs. We develop here a formalism which introduces a
generating function for the prefactor, ie a function
\begin{equation*}
G(k_1, \dots, k_n; y_1, \dots, y_p)
\end{equation*}
where each $y_i$ is associated to one internal propagator
of the graph. For $y_i \equiv 0$ the function $G$ is equal
to the usual correlation function without prefactor
\begin{equation*}
G(k_1, \dots, k_n; 0, \dots, 0)
= G(k_1, \dots, k_n)
= \int \dif[\ell d]p^\text{loop} F(k, p^\text{loop}); 
\end{equation*}
here $p^\text{loop}_i$ are $\ell$ loop momentum variables, to be
integrated over. Let $P(q_1, \dots, q_p)$ be a polynomial of the
internal momenta along the propagators (in general we will not assume
that vector indices are contracted completely). If we substitute the
(components of the) $y$-derivatives for the internal momenta $q$ in
the polynomial $P$ and apply it to the generating function
\begin{equation}
\label{eqn:generating}
\left. P\left(i \frac{\partial}{\partial y}\right) G(k; y) \right|_{y = 0}
= \int \dif[\ell d]p^\text{loop} P(q) F(k, p^\text{loop}),
\end{equation}
we should obtain the prefactor given by the polynomial $P(q)$.
For tree graphs, the
momenta $q_i$ are constant; for loop graphs, there is an
integration over loop momenta.

There exists a method in the ``vector charge'' framework
to reproduce the corresponding amplitudes.
It utilizes the insertion of an additional
pair of oppositely charged external legs on each propagator
- a dipole. We are given a Feynman graph
$\mathfrak g(\tau)$ containing various propagators and vertices in the
Schwinger parametrised form with the loop momenta not yet
integrated. Pick a propagator $j$ transporting
momentum $q_j$ from vertex $v_1$ to vertex $v_2$. On $j$, we want to
insert a dipole with strength $\frac{i y_j}{2}$. The dipole is
constructed from two infinitesimally separated sources: Source $s^+_j$
located arbitrarily within
the propagator and source $s^-_j$ at distance $\varepsilon$ away from
$s^+_j$ in the direction of $q_j$ as we define it
(see fig. \ref{fig:dipole}).
\begin{figure}
\centering \includegraphics[width=8cm]{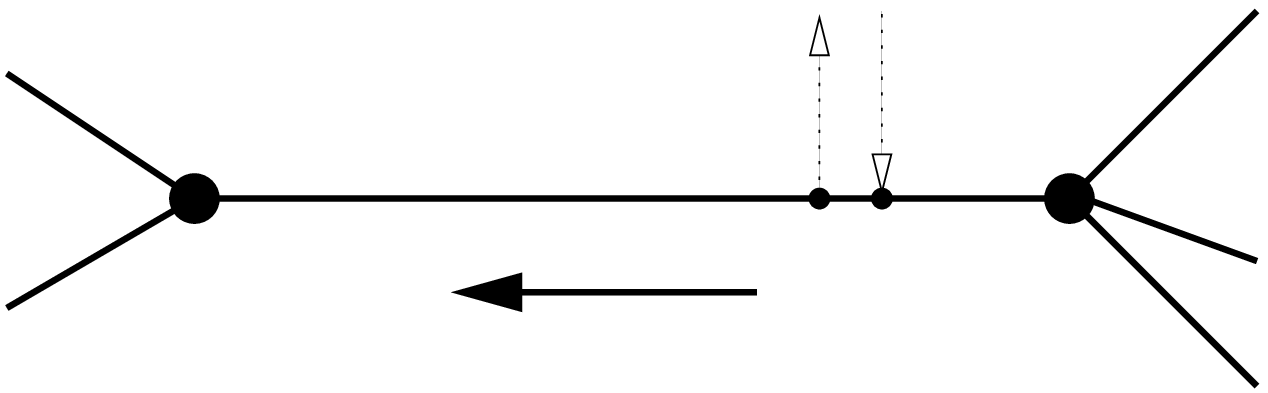}
\put(-120,8){$q_j$}
\put(-86,22){$s_j^-$}
\put(-72,22){$s_j^+$}
\caption{Insertion of a test dipole onto an internal propagator.}
\label{fig:dipole}
\end{figure}
As $\varepsilon$ is supposed to be
infinitesimal, $s^-_j$ will always fit on the propagator.
In the limit $\varepsilon \rightarrow 0$, the position of the dipole
on the propagator becomes $s_j$.
At $s^+_j$, momentum $\frac{iy_j}{2 \varepsilon}$ flows into the
graph; at $s^-_j$, momentum $\frac{iy_j}{2 \varepsilon}$ flows out of the graph;
so the total momentum is conserved. The ``dipole moment'' is the
vector $\mu_j = \varepsilon \cdot \frac{iy_j}{2 \varepsilon} = \frac{iy_j}{2}$.
Writing down the local momentum balance, it is clear that there is
a momentum transfer $q_j + \frac{iy_j}{2 \varepsilon}$
between $s^+_j$ and $s^-_j$. In addition, we multiply the
integrand by a ``dipole self-energy renormalisation constant''
$e^{- \frac{y_j^2}{4 \varepsilon}}$.
Without the dipole, the piece between the insertions naturally
contributes $e^{-\varepsilon q_j^2}$ to the Schwinger parametrised path integral.
The total effect of the dipole including the self-energy renormalisation
is a correction factor
\begin{equation*}
e^{\varepsilon q_j^2
  -\varepsilon (q_j + \frac{iy_j}{2 \varepsilon})^2
  -\frac{y_j^2}{4 \varepsilon}}
= e^{- i y_j \cdot q_j}.
\end{equation*}
This shows that the amplitude is indeed a generating function
for polynomials in $q_j$ by (\ref{eqn:generating}).

We saw that after integration of the loop momenta, the exponent in the Schwinger
integrand (\ref{eqn:G:1}) is equivalent to the potential
of an ensemble of vector charges defined by the external momenta.
The dipole insertions fit quite naturally into this picture.
Because they can be assembled from ``elementary'' vector charges,
the available ingredients of the formalism are completely
sufficient to handle them. All we have to do is calculate the additional potential
terms which arise due to the insertion of the dipoles (for
each propagator). It is advantageous to compute
the dipole-scalar and dipole-dipole pair potentials
before summing over all pairs. We define
\begin{eqnarray}
\label{def:dip:scal}
\varphi^\text{ds}(s_j, v_i)
&=& \lim_{\varepsilon \rightarrow 0}\;
    \frac{1}{\varepsilon} \left( \varphi(s^+_j, v_i)
    - \varphi(s^-_j, v_i) \right) 
 =  - \frac{\partial}{\partial s_j} \varphi(s_j, v_i)\\
\varphi^\text{dd}(s_j, s_i)
\label{def:dip:dip}
&=& \lim_{\varepsilon \rightarrow 0}\;
     \frac{1}{\varepsilon^2} \left( \varphi(s^+_j, s^+_i)
    - \varphi(s^-_j, s^+_i) - \varphi(s^+_j, s^-_i)
    + \varphi(s^-_j, s^-_i) \right) \nonumber\\
&=& - \frac{\partial}{\partial s_i} \varphi^\text{ds}(s_j,s_i)
 =  \frac{\partial^2}{\partial s_j \partial s_i} \varphi(s_j,s_i).
\end{eqnarray}
where the last equation is valid only for $i \neq j$ (the limits
$\varepsilon \rightarrow 0$ are allowed because it is easy to
see that both potentials are of order $\mathcal O(\varepsilon^0)$
and the higher-order terms are not relevant). The
dipole-scalar energy is
$\mu_i \cdot k_j \varphi^\text{ds}(s_j, v_i)$, and similarly for
dipole-dipole. For the self-interaction of a dipole, we define
\begin{equation}
\varphi^\text{d-self}_\varepsilon(s_j)
= - \frac{1}{\varepsilon^2} \varphi(s^-_j, s^+_j).
\end{equation}
The full generating function is then obtained by inserting the factor
\begin{align*}
\exp - \Big[& \frac{i}{2} \sum_{\text{prop.s $j$}} \sum_{\text{vert.s $i$}}
     (y_j \cdot k_i) \varphi^\text{ds}(s_j, v_i)
       + \left(\frac{i}{2}\right)^2 \sum_{\text{prop.s $i<j$}}
     (y_i \cdot y_j) \varphi^\text{dd}(s_i, s_j)\\
 &   + \left(\frac{i}{2}\right)^2 \sum_{\text{prop.s $j$}} y_j^2
     \left( \varphi^\text{d-self}_\varepsilon(s_j)
            - \frac{1}{\varepsilon} \right) \Big].
\end{align*}
By construction, this must be independent of $\varepsilon$, and even of the
coordinates $s_j$ of the dipole insertions on the propagators
\footnote{As the field is constant along the propagator.}.

We apply this formula to the loop graph.
We keep the nomenclature of section \ref{sec:spider} introducing the
loop graph and insert the dipoles at the general loop coordinates $s_j$.
The dipoles are oriented along the canonical loop direction.
We introduce the oriented distance function
\begin{equation*}
[ 0, T [\; \ni \tau_{ji} = (s_j - v_i) \mod T.
\end{equation*}
The relevant potentials are then
\begin{eqnarray}
\label{eqn:potentials}
\varphi(v_j, v_i)
&=& - \frac{\tau_{ji}(T - \tau_{ji})}{T}, \\
\varphi^\text{ds}(s_j, v_i)
&=& - \frac{\partial}{\partial s_j} \frac{- \tau_{ji} (T - \tau_{ji})}{T}
 =  \frac{T - 2\tau_{ji}}{T}, \nonumber \\
\varphi^\text{dd}(s_j, s_i)
&=& - \frac{\partial}{\partial s_i} \frac{T - 2\tau_{ji}}{T}
 =  -\frac{2}{T}, \nonumber \\
\varphi^\text{d-self}_\varepsilon(s_j)
&=& \frac{1}{\varepsilon} \frac{T - \varepsilon}{T}. \nonumber
\end{eqnarray}
The dipole self-energy is divergent as $\varepsilon \rightarrow 0$;
however if we include the additional factor
$e^{-\frac{y_j^2}{4\varepsilon}}$ in the action, then the
``renormalised'' self-energy
\begin{equation*}
\varphi^\text{d-self(ren)}(s_j)
= \lim_{\varepsilon \rightarrow 0}\;
  \left( \varphi^\text{d-self}_\varepsilon(s_j) - \frac{1}{\varepsilon} \right)
= - \frac{1}{T}
\end{equation*}
has a reasonable limit for vanishing dipole extension. The complete
generating factor is therefore
\begin{align*}
\exp & \left\{ - \frac{i}{2} \sum_{i,j=1}^n
     (k_i \cdot y_j) \frac{T - 2 \tau_{ji}}{T}
     - \frac{1}{4T} \Big( \sum_{j=1}^n y_j \Big)^2 \right\}.
\end{align*}
The generating function is therefore
\begin{multline*}
G_\text{1-loop}(k_1, \dots, k_n; y_1 , \dots, y_n)
= (2\pi)^d \delta^{(d)} \Big( \sum_j k_j \Big)
\left( - c_3 \right)^n
\int_0^\infty \frac{\dif{T}}{2T} \left( \frac{1}{4\pi T} \right)^{d/2}\\
\exp \left( - m^2 T \right)
\int_0^T \dif[n]t \exp \left\{ \sum_{i<j} (k_i \cdot k_j)
   \frac{| t_j - t_i | (T - |t_j - t_i| )}{T} \right\}\\
\exp \left\{ - \frac{i}{2} \sum_{i,j=1}^n (y_j \cdot k_i) \frac{T - 2 \tau_{ji}}{T}
     -  \frac{1}{4T} \Big( \sum_{j=1}^n y_j \Big)^2 \right\},
\end{multline*}
Due to the fact that we chose the test dipoles to be imaginary, the
last factor is well-behaved for real $y_j$ when integrating the modulus
$T$. Note, however, that there is an essential singularity at
$\sum_j y_j = 0$ when we allow complex $y_j$. As we compute higher
moments of the propagator momenta, the superficial degree of divergence
of the Feynman amplitude increases until the amplitude needs to be
regularised in order to converge; in terms of the generating function
formalism, this implies that higher derivatives of the generating function
are divergent at the origin. A way out offering itself almost naturally
in formalisms embracing Schwinger parametrisation is dimensional
regularisation (see eg \cite{deCalan:1980qj, Bergere:1974zh}).

\paragraph{Tree graphs.}

On tree graphs, one computes that the generating factor agrees
precisely with the expected form
\begin{equation*}
\exp -i \sum_{\text{prop. j}} y_j \cdot q_j.
\end{equation*}
There are no quadratic terms.

\paragraph{One-loop two-point function.}

This is the simplest non-trivial example. The diagram consists of a
loop with two insertions, connected by two propagators which we will
call left (1) and right (2). We choose as moduli the
lengths $\tau_1$ and $\tau_2$ of the left and right branch. Rather than keeping
two external momenta and imposing momentum conservation,
we assume that there is one momentum $k$ entering and leaving the loop. The
momenta $q_1$ and $q_2$ are defined parallel to the canonical loop
coordinate running around the loop and starting at the insertion where
the external momentum $k$ enters the loop. The dipole on the left
branch is at coordinate $s_1$, on the right branch at coordinate
$s_2$. The generating factor is then (with $T = \tau_1 + \tau_2$)
\begin{multline*}
e^{ - \frac{i}{2} \big[ (y_1 \cdot k) \frac{T - 2 s_1}{T}
     +  (y_2 \cdot k) \frac{T - 2 s_2}{T}  
     -  (y_1 \cdot k) \frac{T - 2 (s_1 + \tau_2)}{T}  
     -  (y_2 \cdot k) \frac{T - 2 (s_2 - \tau_1)}{T} \big]  
    - \frac{1}{4T} \left( y_1 + y_2 \right)^2} \\
= e^{ - i \frac{(y_1 \cdot k) \tau_2 - (y_2 \cdot k) \tau_1}{T}
     - \frac{1}{4T} \left( y_1 + y_2 \right)^2}.
\end{multline*}
We obtain the generating function (without coupling constants)
\begin{align}
\label{eqn:gen2}
G_\text{1-loop}(k; y_1, y_2)
=& \frac{(2\pi)^d}{2}
   \int_0^\infty \dif{\tau_1} \int_0^\infty \dif{\tau_2}
   \left( \frac{1}{4\pi (\tau_1 + \tau_2)} \right)^{d/2}
   e^{ - m^2 (\tau_1 + \tau_2)} \nonumber\\
 & e^{ - \frac{\tau_1 \tau_2}{\tau_1 + \tau_2} k^2
       - i \frac{y_1 \tau_2 - y_2 \tau_1}{\tau_1 + \tau_2} \cdot k
       - \frac{1}{4(\tau_1 + \tau_2)} ( y_1 + y_2 )^2}.
\end{align}

\comment
{
\paragraph{Graphical representation.}

Consider once more the generating rule (\ref{eqn:generating}). We can
see that either, a $y$-derivative acts on the terms in the exponent
linear in $y$ - then we get a prefactor proportional $k_i$ - or it acts
on the quadratic part; then we need another $y$-derivative to act on
the second factor of $y$, in order for the contribution to survive
once we set $y=0$. This enables us to perform a graphical expansion
for the correlations: Namely, derivatives with respect to the
$y_j$ are drawn as lines crossing the propagator $j$; if a derivative
falls on the linear term, then the line ends at the vertex of external
momentum $k_i$; if the derivative falls on the quadratic term, then
the two derivatives are connected by a line. On the other hand, the
$y$-derivatives are contracted finally in the polynomial $P$, and this
can be graphically represented by lines connecting the external ends
of the $y$-lines. If all indices are saturated, then there will be
many $y$-lines ending and beginning at the external momenta $k_j$, and
also closed index loops. [EXAMPLE]
}

\paragraph{Relation to pre-integrated Schwinger representation.}

Of course, there is no magic in here. It is instructive to inspect the
Fourier transform of the generating function
$G_\text{1-loop}(k; y_1, y_2)$ with respect to the variables $y_1$,
$y_2$. This will yield a kernel whose moments are equal to the
derivatives of the generating function at $y_1 = y_2 = 0$; we have
\begin{multline*}
\left. (i\partial_1)^{n_1} (i\partial_2)^{n_2}
G_\text{1-loop}(k; y_1, y_2) \right|_{y_1 = y_2 = 0}\\
= (2\pi)^{-2d} \int \dif[d]q_1 \dif[d]{q_2} q_1^{n_1}\, q_2^{n_2}
  \int \dif[d]y_1 \dif[d]{y_2} e^{i q_1 \cdot y_1 + i q_2 \cdot y_2}
  G_\text{1-loop}(k; y_1, y_2).
\end{multline*}
One finds with (\ref{eqn:gen2})
\begin{multline*}
  (2\pi)^{-2d} \int \dif[d]y_1 \dif[d]{y_2} e^{i q_1 \cdot y_1 + i q_2 \cdot y_2}
  G_\text{1-loop}(k; y_1, y_2) \\
= \frac{1}{2} \delta^{(d)}(k + q_2 - q_1)
  \int_0^\infty \dif{\tau_1} \int_0^\infty \dif{\tau_2}
  e^{ - \tau_1 (q_1^2 + m^2) - \tau_2 (q_2^2 + m^2)}.
\end{multline*}
This is hardly a surprising result; the generating function is
nothing more than the Fourier transform of the Schwinger parametrised
amplitude, \textit{before} integrating out the moments along the
propagators. In the general case, this gives us a convenient way to
obtain the generating functional.

\section{Summary and outlook}

We have demonstrated an ``integrated'' approach to the
Schwinger parametrisation of connected Feynman amplitudes, encompassing the
graph as an entity (world graph) by introducing specific boundary
conditions at the vertices rather than breaking the graph at these
points. With the necessary caveats of interpretation, it can be viewed
as a diffusion process of splitting and re-fusing particles.
The summation over Feynman graphs in order to obtain the total amplitude
is thereby converted into a summation over equivalence classes of
graphs. The Schwinger parameters are generalised to ``Schwinger
moduli'' in the cells of moduli space corresponding to the equivalence
classes. The vertices turn out to be by no means special points on the
graph.

The prefactor $Z_0(\mathfrak g(\tau))$ providing the necessary
normalisation is difficult to determine in general (although in
specific cases, like the spider graph, it is clear). It would be
desirable to find a closed formula for the
determinant of the graph Laplacian, bringing us in a position to
determine explicitly $Z_0(\mathfrak g(\tau))$.

Renormalisation is naturally included into this scheme by the assumption
that the cancellation of infinities is local in moduli space. For this
idea, the concept of a cell complex structure of moduli space
is critical. As long as amplitudes are not renormalised explicitly,
the assumption of non-integer space dimensions is a solution adapted very well
to the formalism.

The analogy to a system of charges on a graph which is a crucial step
in the proof offers a simple way to include vector and tensor
particles by way of a generating function formalism.
The stress again lies not on the computational advantages
of the scheme, but rather on the conceptual side: It is not necessary to
sprinkle new terms over the path integral formula which do not have an
intrinsic meaning. Rather, the generating function has a natural
interpretation as containing test dipoles.

\ed On the speculative side, we suggest the possibility of transforming by
integration by parts the integral over the total moduli space cell
complex associated to a correlation into a sum of integral
contributions from minimal and maximal cells only. This might be of
special interest in the treatment of non-Abelian gauge theories.

\ed Looking further, the charge formalism carries the promise of a simple
treatment of graphs containing particles carrying a
``real world'' electric charge. Such particles would bring along a
cloud of photons coupling to the graph by a derivative coupling; the
external photons would be represented by dipoles on the graph
in the ``charge formalism''. Consequently, these dipoles would
\textit{shield} the pair potential along the world graph. In a similar
vein, it might be advantageous to examine multi-particle production
(which would presumably make the potential on the graph
slightly random, like in a charged, grainy background medium).

The main motivation behind this work lies not in such ``classical''
issues, however; we hope that we might contribute to the rephrasing of
the Schwinger parametrised perturbation amplitude into a bulk
amplitude in the framework of the \AdSCFT correspondence.

\section{Acknowledgements}

The author wants to thank K.-H.~Rehren for thoroughly commenting on
this text, and pointing out where clarifications or more detailed
explanations were necessary. Also thanks to W.~Dybalski for listening
patiently. This work has been financed by the DFG
project RE 1208/4-1 under the supervision of K.-H.~Rehren; the author
also thanks DFG for support for visiting a string school in Trieste, and the 
Albert-Einstein-Institute in Potsdam for their support when visiting
the String Steilkurs I + II. Finally, I want to thank the anonymous
referee whose accurate work has helped greatly to improve the
quality of this publication.


%% file: symmetry.pstex_t
\begin{picture}(0,0)%
\includegraphics{symmetry.eps}%
\end{picture}%
\setlength{\unitlength}{3158sp}%
\begingroup\makeatletter\ifx\SetFigFont\undefined%
\gdef\SetFigFont#1#2{%
  \fontsize{#1}{#2pt}%
  \selectfont}%
\fi\endgroup%
\begin{picture}(5379,1716)(855,-1224)
\put(2249,298){\makebox(0,0)[lb]{\smash{{\SetFigFont{10}{12.0}{\color[rgb]{0,0,0}b.i)}%
}}}}
\put(3696,304){\makebox(0,0)[lb]{\smash{{\SetFigFont{10}{12.0}{\color[rgb]{0,0,0}ii)}%
}}}}
\put(5133,309){\makebox(0,0)[lb]{\smash{{\SetFigFont{10}{12.0}{\color[rgb]{0,0,0}iii)}%
}}}}
\put(870,304){\makebox(0,0)[lb]{\smash{{\SetFigFont{10}{12.0}{\color[rgb]{0,0,0}a.}%
}}}}
\end{picture}%

%% file: twoloop.pstex_t
\begin{picture}(0,0)%
\includegraphics{twoloop.eps}%
\end{picture}%
\setlength{\unitlength}{2763sp}%
\begingroup\makeatletter\ifx\SetFigFont\undefined%
\gdef\SetFigFont#1#2{%
  \fontsize{#1}{#2pt}%
  \selectfont}%
\fi\endgroup%
\begin{picture}(8323,2722)(594,-2621)
\put(1238,-1922){\makebox(0,0)[lb]{\smash{{\SetFigFont{8}{9.6}{\color[rgb]{0,0,0}$\tau_2$}%
}}}}
\put(1859,-1922){\makebox(0,0)[lb]{\smash{{\SetFigFont{8}{9.6}{\color[rgb]{0,0,0}$\tau_3$}%
}}}}
\put(649,-1922){\makebox(0,0)[lb]{\smash{{\SetFigFont{8}{9.6}{\color[rgb]{0,0,0}$\tau_1$}%
}}}}
\put(5341,-1342){\makebox(0,0)[lb]{\smash{{\SetFigFont{12}{14.4}{\color[rgb]{0,0,0}$m_3$}%
}}}}
\put(4308,-821){\makebox(0,0)[lb]{\smash{{\SetFigFont{12}{14.4}{\color[rgb]{0,0,0}$m_2'$}%
}}}}
\put(4317,-1855){\makebox(0,0)[lb]{\smash{{\SetFigFont{12}{14.4}{\color[rgb]{0,0,0}$m_2$}%
}}}}
\put(3182,-1334){\makebox(0,0)[lb]{\smash{{\SetFigFont{11}{13.2}{\color[rgb]{0,0,0}$v_1$}%
}}}}
\put(4233,-1322){\makebox(0,0)[lb]{\smash{{\SetFigFont{11}{13.2}{\color[rgb]{0,0,0}$v_2$}%
}}}}
\put(4027,-2554){\makebox(0,0)[lb]{\smash{{\SetFigFont{11}{13.2}{\color[rgb]{0,0,0}$v_4$}%
}}}}
\put(4001,-59){\makebox(0,0)[lb]{\smash{{\SetFigFont{12}{14.4}{\color[rgb]{0,0,0}$v_3$}%
}}}}
\put(1931,-1015){\makebox(0,0)[lb]{\smash{{\SetFigFont{12}{14.4}{\color[rgb]{0,0,0}$k_1$}%
}}}}
\put(2287,-648){\makebox(0,0)[lb]{\smash{{\SetFigFont{12}{14.4}{\color[rgb]{0,0,0}$k_2$}%
}}}}
\put(2221,-134){\makebox(0,0)[lb]{\smash{{\SetFigFont{12}{14.4}{\color[rgb]{0,0,0}(a)}%
}}}}
\put(2956,-411){\makebox(0,0)[lb]{\smash{{\SetFigFont{12}{14.4}{\color[rgb]{0,0,0}$m_1'$}%
}}}}
\put(3796,-1792){\makebox(0,0)[lb]{\smash{{\SetFigFont{11}{13.2}{\color[rgb]{0,0,0}$t_2$}%
}}}}
\put(3044,-2287){\makebox(0,0)[lb]{\smash{{\SetFigFont{11}{13.2}{\color[rgb]{0,0,0}$m_1$}%
}}}}
\put(3396,-2052){\makebox(0,0)[lb]{\smash{{\SetFigFont{11}{13.2}{\color[rgb]{0,0,0}$t_1$}%
}}}}
\put(7882,-1363){\makebox(0,0)[lb]{\smash{{\SetFigFont{12}{14.4}{\color[rgb]{0,0,0}$m_2$}%
}}}}
\put(8902,-1363){\makebox(0,0)[lb]{\smash{{\SetFigFont{12}{14.4}{\color[rgb]{0,0,0}$m_3$}%
}}}}
\put(6520,-314){\makebox(0,0)[lb]{\smash{{\SetFigFont{12}{14.4}{\color[rgb]{0,0,0}$m_1$}%
}}}}
\put(6141,-109){\makebox(0,0)[lb]{\smash{{\SetFigFont{12}{14.4}{\color[rgb]{0,0,0}(b)}%
}}}}
\put(5694,-791){\makebox(0,0)[lb]{\smash{{\SetFigFont{12}{14.4}{\color[rgb]{0,0,0}$k_2$}%
}}}}
\put(5694,-2002){\makebox(0,0)[lb]{\smash{{\SetFigFont{12}{14.4}{\color[rgb]{0,0,0}$k_1$}%
}}}}
\put(6819,-969){\makebox(0,0)[lb]{\smash{{\SetFigFont{11}{13.2}{\color[rgb]{0,0,0}$v_2$}%
}}}}
\put(6874,-1839){\makebox(0,0)[lb]{\smash{{\SetFigFont{11}{13.2}{\color[rgb]{0,0,0}$v_1$}%
}}}}
\put(7603,-2545){\makebox(0,0)[lb]{\smash{{\SetFigFont{11}{13.2}{\color[rgb]{0,0,0}$v_4$}%
}}}}
\put(7580,-54){\makebox(0,0)[lb]{\smash{{\SetFigFont{11}{13.2}{\color[rgb]{0,0,0}$v_3$}%
}}}}
\put(6571,-2337){\makebox(0,0)[lb]{\smash{{\SetFigFont{12}{14.4}{\color[rgb]{0,0,0}$m_1$}%
}}}}
\put(6093,-1363){\makebox(0,0)[lb]{\smash{{\SetFigFont{12}{14.4}{\color[rgb]{0,0,0}$m_1'$}%
}}}}
\put(6699,-1372){\makebox(0,0)[lb]{\smash{{\SetFigFont{11}{13.2}{\color[rgb]{0,0,0}$t_1'$}%
}}}}
\put(7165,-2176){\makebox(0,0)[lb]{\smash{{\SetFigFont{11}{13.2}{\color[rgb]{0,0,0}$t_1$}%
}}}}
\end{picture}%